\documentclass[11pt]{amsart}
\usepackage{amsfonts}
\usepackage{amsmath}
\usepackage{amsthm}
\usepackage{url}
\usepackage[all]{xy}
\usepackage{latexsym}
\usepackage{amssymb}
\usepackage[cp850]{inputenc}
\usepackage{amsfonts}
\usepackage{graphicx}
\usepackage{dsfont}
\usepackage{float}
\usepackage[table]{xcolor}
\usepackage{pgf}
\usepackage[toc,page]{appendix}
\usepackage[colorlinks=true, urlcolor=blue, citecolor=black, linkcolor=black]{hyperref}

\newtheorem{sat}{Theorem}[section]		
\newtheorem{lem}{Lemma}[section]
\newtheorem{kor}[lem]{Corollary}			
\newtheorem{prop}[lem]{Proposition}
\newtheorem{defi}{Definition}
\newtheorem*{defi*}{Definition}			
\newtheorem*{bei*}{Example}
\newtheorem*{sat*}{Theorem}				
\newtheorem*{kor*}{Corollary}
\newtheorem*{rmk*}{Remark}				
	
\newtheorem*{quest*}{Question}

\let\ssection=\section
\renewcommand{\section}{\setcounter{equation}{0}\ssection}

\newtheorem*{namedtheorem}{\theoremname}
\newcommand{\theoremname}{testing}
\newenvironment{named}[1]{\renewcommand{\theoremname}{#1}\begin{namedtheorem}}{\end{namedtheorem}}

\theoremstyle{remark}
\newtheorem*{bem}{Remark}

\newtheorem*{namedtheoremr}{\theoremnamer}
\newcommand{\theoremnamer}{testing}

			\newcommand{\BH}{\mathbb H}
\newcommand{\BR}{\mathbb R}			
\newcommand{\BN}{\mathbb N}			
\newcommand{\BS}{\mathbb S}			\newcommand{\BZ}{\mathbb Z}
			\newcommand{\BT}{\mathbb T}

\newcommand{\BB}{\mathbb B}

		\newcommand{\CN}{\mathcal N}
		
		\newcommand{\CR}{\mathcal R}
		
		\newcommand{\CV}{\mathcal V}

\newcommand{\D}{\partial}

\DeclareMathOperator{\vol}{vol}		

\DeclareMathOperator{\inj}{inj}

\newcommand{\comment}[1]{}

\DeclareMathOperator{\syst}{syst}
\DeclareMathOperator{\Vol}{Vol}

\DeclareMathOperator{\Prob}{Prob}

\DeclareMathOperator{\proj}{proj}

\DeclareMathOperator{\Corr}{Corr}
\DeclareMathOperator{\Var}{Var}
\DeclareMathOperator{\cov}{cov}

\DeclareMathOperator{\gap}{gap}

\DeclareMathOperator{\round}{Round}

\newcommand{\fsubd}{\mathrel{{\scriptstyle\searrow}\kern-1ex^d\kern0.5ex}}
\newcommand{\bsubd}{\mathrel{{\scriptstyle\swarrow}\kern-1.6ex^d\kern0.8ex}}

\newcommand{\Area}{\mathrm{Area}}

\renewcommand{\epsilon}{\varepsilon}
\renewcommand{\le}{\leqslant}
\renewcommand{\ge}{\geqslant}

\begin{document}

\title[]{Effective estimation of the dimension of a manifold from random samples}
  \author{Lucien Grillet}
\address{UNIV RENNES, IRMAR - UMR 6625, F-35000 RENNES, FRANCE}
\email{grillet.lucien@gmail.com}
\author{Juan Souto}
\address{UNIV RENNES, CNRS, IRMAR - UMR 6625, F-35000 RENNES, FRANCE}
\email{jsoutoc@gmail.com}

\begin{abstract}
We give explicit theoretical and heuristical bounds for how big does a data set sampled from a reach-$1$ submanifold $M$ of euclidian space need to be, to be able to estimate the dimension of $M$ with $90\%$ confidence.
\end{abstract}

\maketitle

\section{Introduction}

The {\em manifold hypothesis} asserts that naturally ocurring data sets $X\subset\BR^s$ behave as if they had been sampled from a low-dimensional submanifold $M$. There is then a large literature on {\em manifold learning}, that is on understanding how properties of the underlying manifold $M$ can be "learned" from the data set $X$---see for example  \cite{BNS,Fefferman,NSW,NM,BG14,LVbook,BKSW}---, but it is maybe fair to say that estimating the dimension $\dim(M)$ of the manifold, that is figuring out the {\em intrinsic dimension} of the data set, seems to be of particular interest. This is maybe so because the dimension is one of the most basic quantities associated to a manifold, but maybe also because of the importance of understanding the applicability to the data set $X$ of different dimension reduction schemes. In any case, there are very numerous estimators for the intrinsic dimension of a data set. We refer to the surveys \cite{Camastra,CamastraS,BKSW} for a brief discussion of many of those estimators and to the references therein for details---see also Chapter 3 in the monograph \cite{LVbook}. Anyways, all these dimension estimators are based on the idea that manifolds look locally like their tangent spaces, that is like euclidean space. For example, non-linear PCA aims at finding the dimension of tangent spaces. ANOVA aims at computing the average angle between vectors in a tangent space, exploiting the fact that for euclidean space itself one can read the dimension from the expected value. The estimator \eqref{eq corrsum intro} below, as well as closely related algorithms due to Takens \cite{Takens83,Takens85}, Theiler \cite{Theiler90} or Grassberger-Procaccia \cite{Grassberger-Procaccia}, exploits the fact that the {\em correlation dimension} of the manifold is just its dimension. There are other such estimators where the correlation dimension is replaced by the {\em Box counting dimension} or the {\em Kolmogorov capacity}. This list of estimators is far from being all inclusive.

All those estimators work very well---in the sense that they yield with high probability the desired result---as long as we apply them with common sense and under suitable conditions, that is if 
\begin{itemize}
\item[1)] we work at a scale at which the manifold $M$ resembles euclidean space, and
\item[2)] the data set $X$ is large enough and the points therein have been obtained by uniform sampling.
\end{itemize} 
However, amazingle little is known when it comes to quantifying 1) or 2). Indeed, there are classical estimates due to Grassberger \cite{Grassberger}, Procaccia \cite{Procaccia}, and Eckmann-Ruelle \cite{Eckmann-Ruelle} giving absolute lower bounds, in terms of the dimension, for the number of measurements needed for the results to be reliable: Grassberger argues that there are absolute lower bounds at all, Procaccia argues that if the intrindic dimension is $\dim$ then one might need a data set of at least $10^{\dim}$ measurements, while Eckmann-Ruelle suggest a lower bound of the form $C\cdot 10^{\frac 12\dim}$ for some large but undetermined $C$. One can summarize these results as asserting that, when the dimension grows, the minimal cardinality of a data set allowing to compute the dimension grows exponentially---this is indeed already the case (see \cite{Weinberger}) when one just wants to distinguish between spheres of consecutive dimensions. 

The absolute lower bounds for the data size that we just mention do not help however with the following less philosophical question: {\em I suspect, or hypothesize, that my data set is sampled out of a manifold with this or that properties. To what extent can I trust the intrinsic dimension estimation given by this or that estimator?} The only result we know along those lines is due to Niyogi-Smale-Weinberger \cite{NSW}, at least as long as one wants to allow for variable curvature manifolds. In \cite{NSW} the authors give namely  an algorithm to compute the homology of a closed submanifold $M\subset\BR^s$ out of a set $X$ sampled from $M$, and they estimated how large does the data size have to be so that their algorithm has success in at least $90\%$ of the cases---observe that knowing the homology we also know the dimension of the manifold. A problem is that the needed data size is astronomical. For example, the estimate in \cite{NSW} for the number of points needed to be sampled to compute with $90\%$ probability of success the homology of the 4-dimensional Clifford torus $\BT^4=\BS^1\times\BS^1\times\BS^1\times\BS^1$ is 24.967.788 points, that is about 25 million points. 

One might contend that the Niyosi-Smale-Weinberger algorithm computes something much more sophisticated than the dimension, that they are basically learning the whole manifold, and that being able to do that is an overkill if what one wants to do is to estimate its dimension. We agree with that point of view. We will consider the estimator 
\begin{equation}\label{eq corrsum intro}
\dim_{\Corr(\epsilon_1,\epsilon_2)}(X)=\round\left(\frac{\log\vert PX(\epsilon_1)\vert-\log\vert PX(\epsilon_2)\vert}{\log(\epsilon_1)-\log(\epsilon_2)}\right)
\end{equation}
where
\begin{equation}\label{eq close pairs}
PX(\epsilon)=\big\{\{x,y\}\subset X\text{ with }0<\vert x-y\vert\le\epsilon\big\},
\end{equation}
and our goal will be to estimate how large does $X\subset M$ need to be for the estimator \eqref{eq corrsum intro} to have a $90\%$ success rate. It is definitively much smaller. For example, applying \eqref{eq corrsum intro} to randomly sampled data sets in $\BT^4$ we have a $90\%$ rate of sucess, as long as we sample $18.262$ points, that is more than 1300 times less than than before.

Evidently, the value of $\dim_{\Corr(\epsilon_1,\epsilon_2)}(X)$ does not only depend on $X$ but also on the chosen scales. In some sense the goal of this paper is to decide how to choose $\epsilon_1$ and $\epsilon_2$ in such a way that one does not need to sample too many points to be $90\%$ sure that \eqref{eq corrsum intro} returns the correct value. This is the kind of results that we will prove:

\begin{sat}\label{sat humane formula for number of points}
For $d=1,\cdots,10$ let $\epsilon_1$ and $\epsilon_2$ be scales as in the table below. Also, given a closed $d$-dimensional manifold $M\subset\BR^s$ with reach $\tau(M)\ge 1$ let $n$ be also as in the following table:
\begin{table}[H]
\begin{tabular}{|l|l|l|l|}
\hline
$d$ & $\epsilon_1$ & $\epsilon_2$ & n \\ \hline
$1$ & $1.5$ & $0.19$ & $9+21\cdot\vol(M)^{\frac 12}$ \\ \hline
 $2$  & $0.78$  & $0.2$ & $94+58\cdot\vol(M)^{\frac 12}$ \\ \hline
 $3$ & $0.63$ & $0.23$ & $635+146\cdot\vol(M)^{\frac 12}$ \\ \hline
 $4$  & $0.54$& $0.23$ & $2786+392\cdot\vol(M)^{\frac 12}$ \\ \hline
 $5$ & $0.46$ & $0.22$ & $7013+1119\cdot\vol(M)^{\frac 12}$ \\ \hline
 $6$  & $0.4$ & $0.21$ & $13221+3366\cdot\vol(M)^{\frac 12}$ \\ \hline
 $7$ & $0.36$ & $0.21$ & $25138+10644\cdot\vol(M)^{\frac 12}$ \\ \hline
 $8$ & $0.33$ & $0.2$ & $50033+34890\cdot\vol(M)^{\frac 12}$ \\ \hline
 $9$ & $0.31$ & $0.19$ & $63876+119533\cdot\vol(M)^{\frac 12}$ \\ \hline
 $10$ & $0.29$ & $0.18$ & $139412+425554\cdot\vol(M)^{\frac 12}$ \\ \hline
\end{tabular}
\end{table}
Then, if we sample independently and according to the riemannian volume form a subset $X\subset M$ consisting of at least $n$ points, then we have 
$$\dim_{\Corr(\epsilon_1,\epsilon_2)}(X)=d$$
with at least $90\%$ probability.
\end{sat}

Here the {\em reach} $\tau(M)$ (see Definition \ref{def reach}), or in the terminology of \cite{NSW} the {\em condition number}, is taken as a measure for the local regularity of the submanifold $M$. It is evident that in order to have specific bounds for the needed data size, we do need to have some a priori control on the local geometry: otherwise we could have a 1-dimesional submanifold so interwoven that it looks as a $d$-dimensional manifold for some $d\ge 2$. Taking the reach as a measure to quantitify to which extent does a submanifold $M\subset\BR^s$ resemble euclidean space seems to be actually pretty common in the field of manifold learning \cite{Aamari,NSW,Fefferman,Boissonat}.

After discussing briefly the correlation dimension and the estimator \eqref{eq corrsum intro} in Section \ref{sec estimators} we discuss some aspects of the geometry of the reach in Section \ref{sec geometry}. We will mostly care about the volume of the thick diagonal 
$$DM(\epsilon)=\{(x,y)\in M\times M\text{ with }\vert x-y\vert\le \epsilon\}$$
For example, in Theorem \ref{sat pain in the ass} we use the Bishop-Gromov theorem and the properties of CAT(1)-spaces to give upper and lower bounds for the ratio $\frac{\vol(DM(\epsilon_1))}{\vol(DM(\epsilon_2))}$ for sufficiently small $\epsilon_1>\epsilon_2$ positive:

\begin{sat}\label{sat pain in the ass}
Suppose that $M\subset\BR^s$ is a $d$-dimensional ($d\ge 1$) closed submanifold with reach $\tau(M)\ge 1$. Then we have
$$\frac{\frac{\epsilon_1}{2}}{\arcsin(\frac {\epsilon_2}2)}\!\left(\!\frac{\sin(\epsilon_1)}{\sin(2\arcsin\frac {\epsilon_2}2)}\!\right)^{d-1}\!\!\!\!\le\frac{\vol(DM(\epsilon_1))}{\vol(DM(\epsilon_2))}\!\le \frac{\int_0^{\sqrt 2\cdot 2\cdot\arcsin(\frac {\epsilon_1}2)}\!\sinh^{d-1}(t)dt}{\int_0^{\sqrt 2\cdot \epsilon_2}\sinh^{d-1}(t)dt}$$
for any two $1>\epsilon_1>\epsilon_2>0$.
\end{sat}

In Section \ref{sec: prob} we come then to core of the present pamphlet. The reason why we care about the volume of the thick diagonal is that for $X\subset M$ we have
$$\frac{\log\vert PX(\epsilon_1)\vert-\log\vert PX(\epsilon_2)\vert}{\log(\epsilon_1)-\log(\epsilon_2)}\sim\frac{\log(\vol(DM(\epsilon_1)))-\log(\vol(DM(\epsilon_2)))}{\log(\epsilon_1)-\log(\epsilon_2)}$$
with large probability, at least if $X$ has been obtained by independently sampling a large number of points according to the riemannian measure on $M$. Basically the goal of Section \ref{sec: prob}, or maybe even the goal of this paper, is to find $\epsilon_1,\epsilon_2>0$ so that a relatively small set $X$ is such that with high probability the left side lies in the interval $(d-\frac 12,d+\frac 12)$. We prove Theorem \ref{sat humane formula for number of points} above in Section \ref{sec: prob}. The scales in the table in the theorem are obtained numerically---the process is also described in Section \ref{sec: prob}, and a computer implementation can be found in \cite{program}.

\subsection*{Heuristical bounds}

It is more or less evident that the bounds given by Theorem \ref{sat humane formula for number of points} are far from being sharp. In Section \ref{sec: heur} we add a heuristical discussion of what could be, in practice, more realistic bounds. The starting point is that the number of sampled points should not be what determines how reliable is the obtained result, but rather the number $\vert PX(\epsilon_1)\vert$ of pairs at our larger scale $\epsilon_1>\epsilon_2>0$. Arguing as if 
\begin{enumerate}
\item at our scales all balls in $M$ were euclidean, and
\item the distances between the two points in pairs as in \eqref{eq close pairs} were independent,
\end{enumerate}
we get that, with $N$ as in Table \ref{table heur bounds intro}, it would suffice to have a data set $X\subset M$ with at least $N$ pairs as in \eqref{eq close pairs} for \eqref{eq corrsum intro} to give the correct answer in about $90\%$ (resp. $70\%$) of the cases.
	\begin{table}[h]
		\begin{tabular}{|l|l|l|l|l|}
			\hline
			$d$ & $\epsilon_1$ & $\epsilon_2$ & {\bf $N$ for 90\%} & {\bf $N$ for 70\%}  \\ \hline
1 & $1.5$ & $0.19$ & $30$ & $10$ \\ \hline
 2  & $0.78$  & $0.2$ & $122$ & $40$ \\ \hline
 3 & $0.63$ & $0.23$ & $249$ & $111$ \\ \hline
 4  & $0.54$& $0.23$ & $516$ & $238$ \\ \hline
 5 & $0.46$ & $0.22$ & $878$ & $360$ \\ \hline
 6  & $0.4$ & $0.21$ & $1329$ & $554$ \\ \hline
 7 & $0.36$ & $0.21$ & $1719$ & $698$ \\ \hline
 8 & $0.33$ & $0.2$ & $2481$ & $1070$ \\ \hline
 9 & $0.31$ & $0.19$ & $3900$ & $1604$ \\ \hline
 10 & $0.29$ & $0.18$ & $5849$ & $2414$ \\ \hline
		\end{tabular}
\medskip
\caption{Heuristic bound $N$ for how large should the cardinality of $PX(\epsilon_1)$ at least be to have $90\%$ (resp. $70\%$) rate of success when using \eqref{eq corrsum intro}.}\label{table heur bounds intro}
	\end{table}

 Note that the assumptions (1) and (2) are not that outlandish, at least if $\epsilon_1$ is small and if the number of pairs is small when compared to the cardinality of $X$. In any case, for what it is worth, numerical simulations (see Table \ref{table expe}) seem to support the values given in Table \ref{table heur bounds intro}. In particular, our numerical simulations also indicate that if $\vert PX(\epsilon_1)\vert$ is less or equal than the value in the right column in Table \ref{table heur bounds intro}, then we should count with about a $30\%$ failure rate when we use \eqref{eq corrsum intro}. 

\begin{bem}	
As a bigger value of $\vert PX(\varepsilon_1)\vert$ results in better performance, it is a natural idea to try to maximize it. This can be done in two ways. First, we can try to obtain more data to increase the number $n$, but as this is not always possible in concrete situations, the second solution is the increase the scale $\varepsilon_1$. The downside of this second method is that considering bigger balls increases the effect of the curvature. One must carefully balance those two effects.
\end{bem}

Thinking on the reliability of \eqref{eq corrsum intro} in terms of the number of pairs has the huge advantage that one does not need to have any a priori knowledge of the volume of the manifold. It is however not clear how different are the heuristic bound from Table \ref{table heur bounds intro} and the formal bound given in Theorem \ref{sat humane formula for number of points}. To be able to compare both bounds we note that, always under the assumption that all $\epsilon_1$ balls in our $d$-manifold $M$ are euclidean, then if we sample a set $X\subset M$ with $n$ points then we expect to have  
$$\vert PX(\epsilon_1)\vert=\frac{n(n-1)}2\cdot\frac{\vol(B^{\BR^d}(\epsilon_1))}{\vol(M)}$$
pairs of points. Using this relation we get for example that, using the heuristic bound, it would suffice to sample $1958$ points from the torus $\BT^4$ for \eqref{eq corrsum intro} to be $90\%$ of the time correct. See Table \ref{table comparisson ns} for more on the comparisson between the heuristic bound and the bound in Theorem \ref{sat humane formula for number of points} for the number of points that suffice to have $90\%$ success rate when using \eqref{eq corrsum intro}. This is, once again, supported by numerical experiments---see Table \ref{table expe 2}. 
\medskip

Note now that thinking of the applicability of \eqref{eq corrsum intro} in terms of the cardinality of $PX(\epsilon_1)$ leads to an implementation of \eqref{eq corrsum intro} which could be applied to data sets $X$ sampled from a manifold $M$ whose reach we ignore. The basic idea is the following: 
\begin{quote}
If we want to check if $M$ has dimension $d$ then we take $\epsilon_1$ minimal so that $\vert PC(\epsilon_1)\vert$ is as large as Table \ref{table heur bounds intro} asks for in dimension $d$, then we choose $\epsilon_2$ so that the ratio $\frac{\epsilon_1}{\epsilon_2}$ is as in Table \ref{table heur bounds intro}), and then compute $\dim_{\Corr(\epsilon_1,\epsilon_2)}(X)$.
\end{quote}
The preceeding discussion implies that for sufficiently rich synthetic data sets this algorithm has a reliability of about 90\%. For the sake of completeness we decided to test it also on data sets each consisting in 200 grayscale pictures of the 3D-model \textit{Suzanne} (see Figure \ref{monkey}) obtained by randomly choosing 1, 2 and 3 Euler angles---the pictures are 64 by 64 pixels large and can thus be represented as points in $\mathbb{R}^{4096}$. See Table \ref{table monkey} for the obtained results.
\begin{figure}[H]
	\centering
	\includegraphics[width=10cm]{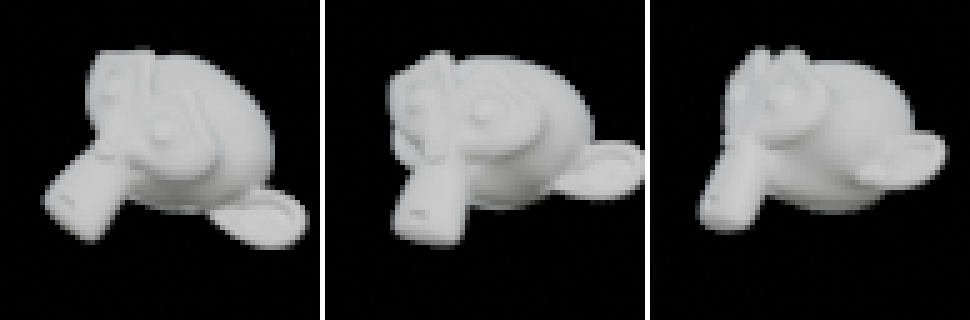}
	\caption{3 sample pictures from the data sets.}
	\label{monkey}
\end{figure}

\subsection*{Comparisson with other estimators}
In Section \ref{sec: compare} we compare briefly the performance of \eqref{eq corrsum intro} with that of a few other dimension estimators: reading the correlation dimension via a log-log plot, ANOVA, and local PCA. The basic observation is that, at least if the volume of $M$ greater or equal to that of the $d$-dimensional Clifford torus and if we take sets of cardinality close to the heuristic bound proposed either in Table \ref{table heur bounds intro} or in Table \ref{table comparisson ns}, then \eqref{eq corrsum intro} seems to perform better than ANOVA and local PCA, at least with the specific implementations we proposed. It is however harder to compare any of these estimators with the incarnation of \eqref{eq corrsum intro} using a log-log plot: it is namely unclear how to meaningfully quantify what does it mean for the non-constant slope of a curve to look constant. It thus only seems to make sense to test the log-log technique numerically, but it is also unclear how to formally do this: while one might well print diagrams for 100 randomly chosen random subsets of a given manifold, it is not clear how can one formally decide what any given diagram is suggesting to us. Still, the limited experiments we conducted lead to the conclussion that the log-log plot method is pretty reliable for relatively small data sets.

\subsection*{Acknowledgements}
The second author would like to thank Kaie Kubjas for getting him interested in this topic.

\section{The estimator}\label{sec estimators}

In this section we recall what is the correlation dimension and how it can be approximated to obtain dimension estimators, but first we introduce some notation that we will use through out the paper. First, distances in eulidean space will be denoted either by $d_{\BR^n}(x,y)$ or by $\vert x-y\vert$. If $x$ and $y$ are points in a Riemannian manifold $M$ then $d_M(x,y)$ is their distance with respect to the Riemannian metric. The ball in $M$ of radius $R$ centered at some point $x\in M$ is denoted by $B^M(x,R)$, although sometimes, when there is no risk of confussion,  we drop the subscript writing thus simply $B(x,R)$. Again, if there is no risk of confussion with the dimension, balls in the ambient euclidean space $\BR^s$ will be denoted by blackboard bold, that is $\BB(x,R)=B^{\BR^s}(x,R)$. Finally, the intersection of balls in the ambient space $\BR^s$ with the manifold $M$ will be denoted by
$$\BB^M(x,R)=M\cap\BB(x,R)$$
With this notation in place, we turn our attention to the correlation dimension and the estimator \eqref{eq corrsum intro}.

As we mentioned earlier, there are plenty of dimension estimators (see \cite{Camastra,CamastraS} and \cite[Chapter 3]{LVbook}). In the language of \cite{CamastraS}, \eqref{eq corrsum intro} is a {\em fractal-based estimator}. Fractal-based because what one is aiming at, is computing a dimension which makes sense for fractal objects---in this case the correlation dimension. Recall that the {\em correlation integral} at scale $\epsilon$ of a Borel measure $\mu$ on $\BR^s$ is the integral
$$C(\mu,\epsilon)=\int\mu(\BB(x,\epsilon))\, d\mu(x)$$
The {\em upper correlation dimension} and {\em lower correlation dimension} are then defined as 
$$D^+(\mu)=\limsup_{\epsilon\to 0}\frac{\log(C(\mu,\epsilon))}{\log(\epsilon)}\text{ and }D^-(\mu)=\liminf_{\epsilon\to 0}\frac{\log(C(\mu,\epsilon))}{\log(\epsilon)}.$$
When both of them agree, then one refers to 
$$D(\mu)=D^+(\mu)=D^-(\mu)$$
as the {\em correlation dimension} of $\mu$. There are numerous situations of interest in dynamics \cite{Simpelaere} in which the correlation dimension of a measure exists. It is also trivial that it exists if $\mu$ is a smooth measure whose support is a submanifold of euclidean space. Since this is the case in which we will find ourselves, we state this fact as a lemma:

\begin{lem}
If $\mu$ is a smooth finite measure of full support of a $d$-dimensional submanifold $M\subset\BR^s$, then the correlation dimension exists and we have $D(\mu)=d$.\qed
\end{lem}

In \cite{Grassberger-Procaccia} Grassberger and Procaccia note that if $\epsilon>0$ is small enough and if we have a set $X=(x_1,\dots,x_n)$ consisting of $n$ points ($n$ very large) sampled independently with respect to the measure $\mu$ then $D(\mu)$ is approximated by
\begin{equation}\label{eq GP}
\dim_{GP}(X,\epsilon)=\frac{\log\left(\frac{\vert DX(\epsilon)\vert}{n(n-1)}\right)}{\log(\epsilon)}
\end{equation}
where 
\begin{equation}\label{eq DX}
DX(\epsilon)=\{(x_i,x_j)\in X\times X\text{ with }i\neq j\text{ and }\vert x_i-x_j\vert\le\epsilon\}
\end{equation}
is the set of (ordered) pairs of points in $X$ within $\epsilon$ of each other. Grassberger and Procaccia propopose \eqref{eq GP} as an estimator for the intrinsic dimension of the data set $X$---variants are discussed by Takens \cite{Takens83, Takens85} and Theiler \cite{Theiler90}. We will however focus here on the version from \cite{LVbook}, or more precisely the quantity 
\begin{equation}\label{eq corrsum sec2}
\dim_{\Corr(\epsilon_1,\epsilon_2)}(X)=\round\left(\frac{\log\vert DX(\epsilon_1)\vert-\log\vert DX(\epsilon_2)\vert}{\log\epsilon_1-\log\epsilon_2}\right)
\end{equation}
Note that this expression is nothing other than \eqref{eq corrsum intro} above.

A problem when using \eqref{eq corrsum sec2} as an estimator, or for that matter when we use \eqref{eq GP}, is that we have to pick up the appropriate scales. In practial implementations, this is often by-passed by taking many possible scales $\epsilon_1<\epsilon_2<\dots<\epsilon_k$, computing $\vert DX(\epsilon_i)\vert$ for each one of those scales, plotting the result in a log-log-plot and choosing what looks to us as the slope at some region where the slope looks constant. In practice, at least when tested on synthetic data, the log-log plot method is surprisingly effective, but it is very unclear how one can get out of that a formal statistical test, or how can one evaluate how much confidence can one have on the obtained number. Theorem \ref{sat humane formula for number of points} from the introduction gives such bounds, at least if the manifold we are working with is sufficiently regular, in the sense of having {\em reach $\tau(M)\ge 1$}. We discuss a few aspects of the geometry of the reach in the next section.

\section{Some Geometry}\label{sec geometry}
In this section we recall a few facts about the geometry of reach-1 submanifolds $M$ of euclidean space. Combining these facts with standard arguments from Riemannian geometry we give upper and lower bounds for the volume of the $\epsilon$ fat diagonal in $M\times M$.

\subsection{Reach-1 manifolds}
We start recalling the definition of the {\em reach} of a closed subset of euclidean space:

\begin{defi}\label{def reach}
The {\em reach} $\tau(S)$ of a closed subset $S\subset\BR^s$ is the supremum of those $T\ge 0$ with the property that for every $x\in\BR^s$ with $d_{\BR^s}(x,S)\le T$ there is a unique point in $S$ closest to $x$. 
\end{defi}

Since the reach was introduced by Federer in \cite{Federer curvature measures}, it has proved to be a useful notion. Indeed, it follows from the very definition that positive reach sets, that is sets $S$ with $\tau(S)>0$, have some neighborhood $\CN(S)$ on which the closest point projection $\pi:\CN(S)\to S$ is well-defined---the existence of such a projection is enough to show that sets of positive reach share many regularity properties with convex sets.

Here we will be working from the very beginning with very regular objects, namely closed smooth submanifolds $M\subset\BR^s$ of euclidean space. In this setting the reach, or rather a lower bound for the reach, helps to quantify how distorted is the inner geometry of $M$ with respect to that of the ambient euclidean space. We summarize what we will need in the following proposition:

\begin{prop}\label{prop basic geometry}
Let $M\subset\BR^s$ be a closed submanifold with reach $\tau(M)\ge 1$. Then we have:
\begin{enumerate}
\item $d_M(x,y)\le 2\arcsin\left(\frac{\vert x-y\vert}{2}\right)$ for any two $x,y\in M$ with $\vert x-y\vert<2$. 
\item The set $\BB^M(x,r)=M\cap\BB(x,r)$ is geodesically convex for any $x\in\BR^s$ and any $r<1$. 
\item We have $\angle\big(\gamma'(0),\gamma'(t)\big)\le d_M(\gamma(0),\gamma(\ell))$ for every geodesic $\gamma:[0,\ell]\to M$ parametrized by arc length. 
\item The manifold $M$ has sectional curvature pinched by $-2\le\kappa_M\le 1$.
\end{enumerate}
\end{prop}

See, in that order Lemma 3, Corollary 1 and Lemma 5 in \cite{Boissonat} for the first three claims of Proposition \ref{prop basic geometry}. See then Proposition A.1 (iii) in \cite{Aamari} for the final claim. In any case, we refer to \cite{Aamari} and \cite{Boissonat}, and to the references therein for general facts about submanifolds $M\subset\BR^s$ of positive reach.

\subsection{Injectivity radius}
Armed with Proposition \ref{prop basic geometry} we can now derive a lower bound for the injectivity radius of those submanifolds $M\subset\BR^s$ with reach $\tau(M)\ge 1$. Recall that the {\em injectivity radius} of a geodesically complete Riemannian manifold $M$ at a point $x\in M$ is defined as
$$\inj(M,x)=\sup\{t>0\ \vert\ \exp_x:T_xM\to M\text{ is injective on }B^{T_xM}(0,t)\}$$
where $B^{T_xM}(0,t)=\{v\in T_xM\text{ with }\vert v\vert\le r\}$ and where $\exp_x$ is the Riemannian exponential map. The {\em injectivity radius} of $M$ itself is then defined to be
$$\inj(M)=\inf_{x\in M}\inj(M,x).$$
Recall also that the {\em systole} $\syst(M)$ of a closed manifold $M$ is the length of the shortest non-trivial closed geodesic. The importance of the systole now is that, together with an upper bound $\kappa_M\le\kappa$ for the sectional curvature, the systole yields very a usable bound \cite[Thm. 89]{Berger} for the injectivity radius:
$$\inj(M)\ge\min\left\{\frac{\pi}{\sqrt\kappa},\frac 12\syst(M)\right\}.$$
Altogether, see \cite{Berger} for basic facts and definitions from Riemannian geometry. 

Suppose now that $M\subset\BR^s$ is a closed submanifold of euclidean space with $\tau(M)\ge 1$, and note the that assumption that $M$ is closed implies that it is metrically complete and thus geodesically complete by the Hopf-Rinow theorem. On the other hand if $\gamma:\BS^1\to M\subset\BR^s$ is any smooth curve and if $t$ is such that $\vert \gamma(t)-\gamma(0)\vert$ is maximal, then $\langle\gamma'(0),\gamma'(t)\rangle=0$. We thus get from (3) in Proposition \ref{prop basic geometry} that every closed geodesic reaches at least distance $\frac\pi 2$, and hence that
$$\syst(M)\ge \pi.$$ 
Now, this fact together with the upper bound for the sectional curvature from (4) in Proposition \ref{prop basic geometry} and with the lower bound for the injectivity radius yields the following:

\begin{kor}\label{kor inj conv}
If $M\subset\BR^s$ is a closed submanifold of euclidean space with reach $\tau(M)\ge 1$, then we have $\inj(M)\ge\frac\pi 2$.\qed
\end{kor}

It should be noted that if $\dim(M)=1$ and $\tau(M)\ge 1$ then $\syst(M)\ge 2\pi$ and hence $\inj(M)\ge\pi$. The standard circle shows that this is optimal. We do not know by how much can one improve the bound given in Corollary \ref{kor inj conv} in other dimensions.

\subsection{Volumes of balls}

Although we will be mostly interested in the volume of sets $\BB^M(x,r)=M\cap\BB(x,r)$ we start by considering the volume of actual metric balls in the manifold. First note that having bounds on the curvature, we also get bounds on the volumes of balls, at least as long as the radius remains below the injectivity radius. More precisely, if $M$ is a Riemannian manifold of dimension $d\ge 2$ and with curvature pinched in $[-2,1]$ then we have
\begin{equation}\label{eq bounds volume balls}
\vol(B^{\BS^d}(r))\le\vol(B^M(x,r))\le\vol(B^{\frac 1{\sqrt 2}\BH^d}(r))
\end{equation}
for all $r<\inj_M(x)$ (see \cite[Thm. 103]{Berger} and \cite[Thm. 107]{Berger}). Here, the sphere $\BS^d$ and the scaled hyperbolic space $\frac 1{\sqrt 2}\BH^d$ are respectively the simply connected complete $d$-dimensional manifolds of constant curvature $1$ and $-2$. Having bounds for the volumes of balls we also have bounds for the ratios between volumes of balls of different radius. We can get however somewhat better results:

\begin{prop}\label{prop volume of balls}
Suppose that $M$ is a $d$-dimensional ($d\ge 2$) Riemannian manifold with sectional curvature pinched in $[-2,1]$. We then have
$$\frac Rr\cdot \left(\frac{\sin(R)}{\sin(r)}\right)^{d-1}\le\frac{\vol(B^M(R,x))}{\vol(B^M(r,x))}\le \frac{\int_0^{\sqrt 2\cdot R}\sinh^{d-1}(t)\ dt}{\int_0^{\sqrt 2\cdot r}\sinh^{d-1}(t)\ dt}$$
for any two $0<r<R<\min\{\pi,\inj(M)\}$ and any $x\in M$.
\end{prop}

\begin{proof}
We begin with the upper bound. We get from the curvature bound $\kappa\ge -2$ and the Bishop-Gromov comparison theorem \cite[p.310]{Berger} that
$$\frac{\vol(B^M(x,R))}{\vol(B^M(x,r))}\le\frac{\vol(B^{\frac 1{\sqrt 2}\BH^d}(R))}{\vol(B^{\frac 1{\sqrt 2}\BH^d}(r))}=\frac{\int_0^{\sqrt 2\cdot R}\sinh^{d-1}(t)\ dt}{\int_0^{\sqrt 2\cdot r}\sinh^{d-1}(t)\ dt}$$
for all $x\in M$, as we wanted.
\medskip

Let us now deal with the lower bound. Well, the fact that $M$ has curvature pinched from above by $1$ implies that $M$ is is locally a CAT($1$)-space---see \cite{Bridson-Haffliger,Cheerger-Ebin} for facts about CAT($\kappa$)-spaces and comparisson geometry. Recall now that we are working at a scale smaller than $\pi$ and the injectivity radius. In particular, the CAT($1$) property implies that geodesic triangles we encounter are thinner in our manifold than in $\BS^d$. This implies in particular that, for $0<r<R<\min\{\pi,\inj(M)\}$, the radial projection
$$\proj:S^M(x,R)\to S^M(x,r),\ \ \proj(y)=\exp_x(r\cdot R^{-1}\cdot\exp_x^{-1}(y))$$
contracts distances more (expands distances less) than the corresponding map in the sphere, meaning that $\proj$ is $\frac{\sin(r)}{\sin(R)}$-Lipschitz. We deduce thus that 
$$\Area(S^M(x,r))\le\left(\frac{\sin(r)}{\sin(R)}\right)^{d-1}\cdot\Area(S^M(x,R))$$
where $S^M(x,r)=\D B^M(x,r)$ is the distance-$r$-sphere in $M$ and $\Area(\cdot)$ stands for the $(d-1)$-dimensional volume. Anyways, if we set $T=\frac rR$ then we get from the co-area formula that
\begin{align*}
\vol(B^M(x,r))
&=\int_0^r\Area(S^M(x,t))\ dt\\
&\stackrel{s=\frac 1Tt}=\int_0^R\Area(S^M(x,T\cdot s))\cdot T\ ds\\
&\le \frac rR\cdot\int_0^R\left(\frac{\sin(Ts)}{\sin(s)}\right)^{d-1}\Area(S^M(x,s))\ ds
\end{align*}
The function 
$$(0,\pi)\to\BR,\ s\to\frac{\sin(Ts)}{\sin(s)}$$ 
is monotonically increasing (because $T=\frac rR\in(0,1)$). This means that 
\begin{align*}
\vol(B^M(x,r))
&\le \frac rR\cdot \left(\frac{\sin(r)}{\sin(R)}\right)^{d-1}\cdot\int_0^R\Area(S^M(x,s))\ ds\\
&=\frac rR\cdot \left(\frac{\sin(r)}{\sin(R)}\right)^{d-1}\cdot\vol(B^M(x,R))
\end{align*}
And we are done.
\end{proof}

We come now to the result we really care about:

\begin{kor}\label{kor volume of balls}
Suppose that $M\subset\BR^s$ is a closed $d$-dimensional ($d\ge 1$) submanifold with reach $\tau(M)\ge 1$. We then have
$$\vol(B^{\BS^d}(r))\le\vol(\BB^M(x,r))\le\vol\left(B^{\frac 1{\sqrt 2}\BH^d}\left(2\arcsin\left(\frac r2\right)\right)\right)$$
and
$$\frac{\frac{R}{2}}{\arcsin(\frac {r}2)}\!\left(\!\frac{\sin(R)}{\sin(2\arcsin\frac {r}2)}\!\right)^{d-1}\!\!\!\!\le\frac{\vol(\BB^M(x,R))}{\vol(\BB^M(x,r))}\!\le \frac{\int_0^{\sqrt 2\cdot 2\cdot\arcsin(\frac {R}2)}\!\sinh^{d-1}(t)dt}{\int_0^{\sqrt 2\cdot r}\sinh^{d-1}(t)dt}$$
for any two $0<r<R<1$ and any $x\in M$.
\end{kor}
\begin{proof}
Let's assume for the time being that $d\ge 2$. From part (1) in Proposition \ref{prop basic geometry} we get for any $t<2$ that 
\begin{equation}\label{eq balls in balls}
B^M(x,t)\subset\BB^M(x,t)\subset B^M\left(x,2\arcsin\left(\frac t2\right)\right)
\end{equation}
The first claim then follows directly from \eqref{eq bounds volume balls}. 

On the other hand, if we combine \eqref{eq balls in balls} with the upper bound in Proposition \ref{prop volume of balls} we get that
\begin{align*}
\frac{\vol(\BB^M(x,R))}{\vol(\BB^M(x,r))}
&\le\frac{\vol\left(B^M\left(x,2\arcsin\left(\frac R2\right)\right)\right)}{\vol(B^M(r,x))}\\
&\le \frac{\int_0^{\sqrt 2\cdot 2\arcsin\left(\frac R2\right)}\sinh^{d-1}(t)\ dt}{\int_0^{\sqrt 2\cdot r}\sinh^{d-1}(t)\ dt},
\end{align*}
and we are done with the upper bound of the second claim. The lower bound is obtained analogously and we leave the details to the reader.
\medskip

So far we have been focusing on the case of dimension $d\ge 2$. In dimension $d=1$ we actually get from \eqref{eq balls in balls} that
$$R-r\le \vol(\BB^M(x,R))-\vol(\BB^M(x,R))\le 2\arcsin\left(\frac{R-r}2\right)$$
This implies directly that 
$$\frac Rr\le\frac{\vol(\BB^M(x,R))}{\vol(\BB^M(x,r))}\le 1+\frac{2\arcsin\left(\frac{R-r}2\right)}r$$
for all $0<r<R<1$. The so-obtained bound for $d=1$ is slightly better than the one we had claimed. We also note that this bound still works for $0<r<R<2$.
\end{proof}

\subsection{Volume of thick diagonal}

Our goal here is to estimate how the volume $\vol(DM(\epsilon))$ of the {\em $\epsilon$-thick diagonal} 
$$DM(\epsilon)=\{(x,y)\in M\times M\text{ with }\vert x-y\vert\le\epsilon\}$$
of a reach-1 submanifold $M\subset\BR^s$ varies when we replace $\epsilon$ by something else. Here the volume is computed as a subset of the Riemannian manifold $M\times M$, but can be expressed as an integral over $M$ as follows:
\begin{equation}\label{eq volume thick diagonal as integral}
\vol(DM(\epsilon))=\int_M\vol(\BB^M(x,\epsilon))\ dx
\end{equation}
In other words, $\vol(DM(\epsilon))$ is nothing other than the correlation integral at scale $\epsilon$ of the riemannian measure of $M$ when considered as a measure on the ambient euclidean space $\BR^s$. We stress that the thick diagonal is defined in terms of the ambient distance in euclidean space, not in terms of the intrinsic distance of $M$. Anyways, now we prove the following:

\begin{sat}\label{sat pain in the ass}
Suppose that $M\subset\BR^s$ is a $d$-dimensional ($d\ge 1$) closed submanifold with reach $\tau(M)\ge 1$. Then we have
$$\frac{\frac{\epsilon_1}{2}}{\arcsin(\frac {\epsilon_2}2)}\!\left(\!\frac{\sin(\epsilon_1)}{\sin(2\arcsin\frac {\epsilon_2}2)}\!\right)^{d-1}\!\!\!\!\le\frac{\vol(DM(\epsilon_1))}{\vol(DM(\epsilon_2))}\!\le \frac{\int_0^{\sqrt 2\cdot 2\cdot\arcsin(\frac {\epsilon_1}2)}\!\sinh^{d-1}(t)dt}{\int_0^{\sqrt 2\cdot \epsilon_2}\sinh^{d-1}(t)dt}$$
for any two $0<r<R<1$.
\end{sat}

\begin{proof}
From the expression \eqref{eq volume thick diagonal as integral} we get that
$$\min_{x\in M}\frac{\vol(\BB^M(R,x))}{\vol(\BB^M(r,x))}\le \frac{\vol(DM(R))}{\vol(DM(r))}\le\max_{x\in M}\frac{\vol(\BB^M(R,x))}{\vol(\BB^M(r,x))}$$
Now the claim follows from Corollary \ref{kor volume of balls}.
\end{proof}

\subsection{The gap}
As we mentioned in the introduction, if we sample more and more points from a manifold $M$ and we apply the our algorithm, then what we are doing is computing the quantity 
$$\frac{\log(\vol(DM(\epsilon_1)))-\log(\vol(DM(\epsilon_2)))}{\log(\epsilon_1)-\log(\epsilon_2)}.$$
Now, armed with Theorem \ref{sat pain in the ass} we could analyse what happens when one of the scales $\epsilon_1$ and $\epsilon_2$ tends to $0$, or when the gap between them tends to $0$, or when the dimension grows. All of this would be nice and well, but what we actually care about is to find scales that on the one hand keep our estimator reliable while working with as few points as possible. Implementing numerically a procedure described in Section \ref{sec: searche decent scales} below, we find convenient scales for the dimensions we are mostly interested in. Note that the condition $R<1$ can be replaced by $R<2$ for dimension 1, as explained in the proof of Corollary \ref{kor volume of balls}.

\begin{kor}\label{kor examples of gaps}
Let $M\subset\BR^s$ be a submanifold of dimension $d=1,2,\dots,10$ and with reach $\tau(M)\ge 1$, and let $\epsilon_1,\epsilon_2$ and $\gap_d$ be as in the table below. Then we have
$$d-\frac 12+\gap_d\le\frac{\log\left(\frac{\vol(\BB^M(\epsilon_1,x))}{\vol(\BB^M(\epsilon_2,x))}\right)}{\log\left(\frac{\epsilon_1}{\epsilon_2}\right)}\le d+\frac 12-\gap_d$$

\begin{minipage}{0.5\textwidth}
\begin{table}[H]
\begin{tabular}[t]{|l|l|l|l|l|}
\hline
$d$ & $\epsilon_1$ & $\epsilon_2$ & $\gap_d$  \\ \hline
$1$ & $1.5$ & $0.19$ & $0.463241$ \\ \hline
$2$  & $0.78$  & $0.2$ & $0.387573$ \\ \hline
$3$ & $0.63$ & $0.23$ & $0.307476$ \\ \hline
$4$  & $0.54$& $0.23$ & $0.249891$ \\ \hline
$5$ & $0.46$ & $0.22$ & $0.223958$ \\ \hline
\end{tabular}
\end{table}
\end{minipage}\hfill
\begin{minipage}{0.5\textwidth}
\begin{table}[H]
\begin{tabular}[t]{|l|l|l|l|l|}
\hline
$d$ & $\epsilon_1$ & $\epsilon_2$ & $\gap_d$  \\ \hline
$6$  & $0.4$ & $0.21$ & $0.208521$ \\ \hline
$7$ & $0.36$ & $0.21$ & $0.178814$ \\ \hline
$8$ & $0.33$ & $0.2$ & $0.166892$ \\ \hline
$9$ & $0.31$ & $0.19$ & $0.155560$ \\ \hline
$10$ & $0.29$ & $0.18$ & $0.152528$ \\ \hline
\end{tabular}
\end{table}
\end{minipage}
\qed
\end{kor}

\section{Sampling the thick diagonal}\label{sec: prob}

In this section we will be still assuming that $M\subset\BR^s$ is a closed $d$-dimensional submanifold ($d\ge 1$) with reach $\tau(M)\ge 1$. Basically our goal is to bound the number of points that we have to sample from $M$ to get a decent result when we use \eqref{eq corrsum intro}.

\subsection{Some probability}
Suppose that we have a symmetric, say bounded, function 
$$f:M\times M\to\BR$$
We are interested in the sequence of random variables
\begin{equation}\label{eq our random variables}
X_n^f:M^{\BN}\to\BR,\ \ X_n(x_1,x_2,\dots)\mapsto\sum_{i,j\le n,\ i\neq j}f(x_j,x_j)
\end{equation}
when $n$ tends to $\infty$. Here we have endowed $M$ with the probability measure 
$$\Prob=\frac 1{\vol(M)}\vol$$ 
proportional to the Riemannian measure. Accordingly, $M^k$ and $M^\BN$ are all endowed with the corresponding product measure, again denoted by $\Prob$. 

Being the sum of random variables, the expectation and variance of $X_n^f$ are easy to get. Here they are:
\begin{align*}
E(X_n^f)&=n(n-1)\cdot E(f)\\
\Var(X_n^f)&=2\cdot n(n-1)\cdot\Var(f)+4\cdot n(n-1)(n-2)\cdot\cov(f)
\end{align*}
where $\cov(f)$ is the co-variance of $(x_1,\dots,x_n)\mapsto f(x_1,x_2)$ and $(x_1,\dots,x_n)\mapsto f(x_1,x_3)$, or in a formula
$$\cov(f)=\int_{M\times M\times M} f(x,y)\cdot f(x,z)\ d\Prob(x,z,y)- E(f)^2.$$
Besides the expectation and the variance, what we will need to estimate is the quantity $\frac{\Var}{E^2}$ for the random variables $X_n^f$. Well, this is what we get if we just use our expressions for the expectation and the variance:
\begin{equation}\label{eq var/er2}
\frac{\Var(X_n^f)}{E(X_n^f)^2}=\frac 2{n(n-1)}\cdot\frac{\Var(f)}{E(f)^2}+\frac{4(n-2)}{n(n-1)}\cdot\frac{\cov(f)}{E(f)^2}
\end{equation}
The reason why we will care about this last quantity is the following surely standard consequence of the  Bienaym\'e-Chebyshev inequality:

\begin{lem}\label{prop log chebysheff}
For any integrable random variable $X$ in a probability space $X$ we have 
$$\Prob\left(\left\vert\log\left(\frac{X}{E(X)}\right)\right\vert>\delta\right)\le \frac 1{(1-e^{-\delta})^2}\cdot \frac{\Var(X)}{E(X)^2}$$
\end{lem}
\begin{proof}
Well, let us compute
\begin{align*}
\Prob\left(\left\vert\log\left(\frac{X}{E(X)}\right)\right\vert>\delta\right)
&=\Prob\left(\frac{X}{E(X)}\notin[e^{-\delta},e^{\delta}]\right)\\
&\le\Prob\left(\left\vert\frac{X}{E(X)}-1\right\vert>1-e^{-\delta}\right)\\
&\le\Prob\left(\left\vert X-E(X)\right\vert>(1-e^{-\delta})\cdot E(X)\right)
\end{align*}
Setting 
$$k=\frac{(1-e^{-\delta})\cdot E(X)}{\sqrt{\Var(X)}}$$ 
in the standard  Bienaym\'e-Chebysheff inequality
$$\Prob\left(\vert X-E(X)\vert\ge k\cdot\sqrt{\Var(X)}\right)\le k^{-2}$$ 
we get 
$$\Prob\left(\left\vert\log\left(\frac{X}{E(X)}\right)\right\vert>\delta\right)\le \frac{\Var(X)}{(1-e^{-\delta})^2\cdot E(X)^2}$$
as we had claimed.
\end{proof}

\subsection{The function we care about}

We are going to be interested in all of this in the case that $f=f_\epsilon$ is the characteristic function of $DM(\epsilon)$, that is
\begin{equation}\label{eq function}
f_\epsilon(x,y)=\left\{\begin{array}{cl}
1 & \text{ if }\vert x-y\vert\le\epsilon \\
0 & \text{ otherwise}
\end{array}\right.
\end{equation}
This function satisfies that 
$$E(f_\epsilon)=\frac{\vol(DM(\epsilon))}{\vol(M)^2}$$
and hence we get that
$$E(X_n^{f_\epsilon})=n(n-1)\cdot\frac{\vol(DM(\epsilon))}{\vol(M)^2}$$
for all $n\ge 2$. Besides knowing the expectation, to be able to use Lemma \ref{prop log chebysheff} when we need to know, or at least estimate, is the quantity $\frac{\Var}{E^2}$. To apply \eqref{eq var/er2} we need first to be able to estimate the variance and covariance of $f_\epsilon$. Well, since $f_\epsilon$ only takes the values $0$ and $1$, the variance is easily calculated:
$$\Var(f_\epsilon)=E(f_\epsilon)-E(f_\epsilon)^2=\frac{\vol(DM(\epsilon))}{\vol(M)^2}-\left(\frac{\vol(DM(\epsilon))}{\vol(M)^2}\right)^2$$
When it comes to the covariance we have
\begin{align*}
\cov(f_\epsilon)
&=\int\left(\frac{\vol(\BB^M(x,\epsilon))}{\vol(M)}\right)^2d\Prob(x)-\frac{\vol(DM(\epsilon))^2}{\vol(M)^4}\\
&=\int\left(\frac{\vol(\BB^M(x,\epsilon))}{\vol(M)}\right)^2d\Prob(x)-\left(\int_M\frac{\vol(DM(\epsilon))}{\vol(M)}d\Prob(x)\right)^2\\
&=\Var\left(x\mapsto \frac{\vol(\BB^M(x,\epsilon))}{\vol(M)}\right)
\end{align*}
This means that when the volume of $\BB^M(x,\epsilon)=M\cap\BB(x,\epsilon)$ is constant then the covariance vanishes. This is for example the case for $M=\BS^d\subset\BR^{d+1}$ or for the Clifford torus $M=\BT^d\subset\BR^{2d}$. However, in general we do not get anything better than the bound coming from Popoviciu's inequality, that is
$$\cov(f_\epsilon)\le\frac{\left(V^M_{\max}(\epsilon)-V^M_{\min}(\epsilon)\right)^2}{4\cdot\vol(M)^2}$$
where we have set
$$V_{\max}^M(\epsilon)=\max_{x\in M}\vol(\BB^M(x,\epsilon))\text{ and }V_{\min}^M(\epsilon)=\min_{x\in M}\vol(\BB^M(x,\epsilon))$$
Now, using \eqref{eq var/er2}, the bound for $\cov(f_\epsilon)$, as well as the bound $\vol(DM(\epsilon))\ge\vol(M)\cdot V_{\min}(\epsilon)$ we get that
$$\frac{\Var(X_n^{f_\epsilon})}{E(X_n^{f_\epsilon})^2}\le 
\frac 2{(n-1)^2}\cdot\frac{\vol(M)}{V^M_{\min}(\epsilon)}+\frac 1{n-1}\cdot\left(\frac{V_{\max}^M(\epsilon)}{V_{\min}^M(\epsilon)}-1\right)^2$$
To get bounds that only depend on the dimension and on $\epsilon<2$ recall that from Corollary \ref{kor volume of balls} we get that
\begin{equation}\label{eq weird notation}
\begin{split}
V^M_{\min}(\epsilon)&\ge \vol(B^{\BS^d}(\epsilon))\stackrel{\text{def}}=\CV(\epsilon)\\
\frac{V_{\max}^M(\epsilon)}{V_{\min}^M(\epsilon)}&\le \frac{\vol\left(B^{\frac 1{\sqrt 2}\BH^d}\left(2\arcsin\left(\frac \epsilon 2\right)\right)\right)}{\vol(B^{\BS^d}(\epsilon))}\stackrel{\text{def}}=\CR(\epsilon)
\end{split}
\end{equation}
Using these bounds we get
$$\frac{\Var(X_n^{f_\epsilon})}{E(X_n^{f_\epsilon})^2}\le\frac 2{(n-1)^2}\cdot\frac{\vol(M)}{\CV(\epsilon)}+\frac 1{n-1}\cdot\left(\CR(\epsilon)-1\right)^2$$
We record what we have so far:

\begin{lem}\label{lem listening to labordeta}
Let $M\subset\BR^s$ be a closed submanifold with dimension $\dim(M)=d$ and reach $\tau(M)\ge 1$, and for some $\epsilon<2$ and $n\in\BN$ consider $f_\epsilon$ and $X_n^{f_\epsilon}$ as in \eqref{eq function} and \eqref{eq our random variables}. Then we have
\begin{align*}
E(X_n^{f_\epsilon})=&n(n-1)\cdot\frac{\vol(DM(\epsilon))}{\vol(M)^2}\\
\frac{\Var(X_n^{f_\epsilon})}{E(X_n^{f_\epsilon})^2}\le&\frac 2{(n-1)^2}\cdot\frac{\vol(M)}{\CV(\epsilon)}+\frac 1{n-1}\cdot\left(\CR(\epsilon)-1\right)^2
\end{align*}
where $\CV(\epsilon)$ and $\CR(\epsilon)$ are as in \eqref{eq weird notation}.
\qed
\end{lem}

Before moving any further let us give explicit formulas for $\CV(\epsilon)$ and $\CR(\epsilon)$:
\begin{equation}\label{eq weird notation2}
\begin{split}
\CV(\epsilon)&=\vol(\BS^{d-1})\cdot\int_0^\epsilon\sin^{d-1}(t)dt\\
\CR(\epsilon)&=\frac{2^{-\frac d2}\int_0^{2\sqrt 2\arcsin\frac\epsilon 2}\sinh(t)^{d-1}dt}{\int_0^\epsilon\sin^{d-1}(t)dt}
\end{split}
\end{equation}

Note also that the two summands in the bound for $\frac{\Var(X_n^{f_\epsilon})}{E(X_n^{f_\epsilon})^2}$ in Lemma \ref{lem listening to labordeta} are rather different. Assume for example that $d$ is fixed. Then the weight of the second factor decreases when $\epsilon$ decreases. On the other hand the value of the first one explodes. Recall also that the second factor can be ignored if $\cov$ vanishes, that is if all balls $\BB^M(x,\epsilon)$ have the same volume.

\subsection{Some technical results}
Recall that to estimate the dimension of $M$ via \eqref{eq corrsum intro}, or equivalently via \eqref{eq corrsum sec2}, what we do is to take, for two scales $\epsilon_1>\epsilon_2$ random values of $X^{\epsilon_1}_n=X^{f_{\epsilon_1}}_n$ and $X^{\epsilon_2}_n=X^{f_{\epsilon_2}}_n$, compute
$$\frac{\log \frac{X^{\epsilon_1}_n}{n(n-1)}-\log\frac{X^{\epsilon_2}_n}{n(n-1)}}{\log\epsilon_1-\log\epsilon_2}$$
and hope the that obtained value has something to do with $\dim(M)$. Well, what we get from Lemma \ref{prop log chebysheff} is an estimate of the probability that this value is far from the expectation. Indeed, since we have 
\begin{align*}
&\Prob\left(\left\vert\frac{\log X^{\epsilon_1}_n-\log X^{\epsilon_2}_n}{\log\epsilon_1-\log\epsilon_2}-\frac{\log \vol(DM(\epsilon_1))-\log \vol(DM(\epsilon_2))}{\log\epsilon_1-\log\epsilon_2}\right\vert>\rho\right)\\
&\phantom{blablabla}=\Prob\left(\left\vert\frac{\log X^{\epsilon_1}_n-\log X^{\epsilon_2}_n}{\log\epsilon_1-\log\epsilon_2}-\frac{\log E(X^{\epsilon_1}_n)-\log E(X^{\epsilon_2}_n)}{\log\epsilon_1-\log\epsilon_2}\right\vert>\rho\right)\\
&\phantom{blablabla}=\Prob\left(\left\vert\log\left(\frac{X^{\epsilon_1}_n}{E(X^{\epsilon_1}_n)}\right)-\log\left(\frac{X^{\epsilon_2}_n}{E(X^{\epsilon_2}_n)}\right)\right\vert>\log\left(\left(\frac{\epsilon_1}{\epsilon_2}\right)^\rho\right)\right)\\
&\phantom{blablabla}\le\sum_{i=1,2}\Prob\left(\left\vert\log\left(\frac{X^{\epsilon_i}_n}{E(X^{\epsilon_i}_n)}\right)\right\vert>\log\left(\left(\frac{\epsilon_1}{\epsilon_2}\right)^{\frac \rho2}\right)\right)\\
&\phantom{blablabla}\le\frac 1{\left(1-\left(\frac{\epsilon_2}{\epsilon_1}\right)^{\frac \rho2}\right)^2}\sum_{i=1,2} 
\frac{\Var(X_n^{\epsilon_i})}{E(X_n^{\epsilon_i})^2}
\end{align*}
Plugging in the statement of Lemma \ref{lem listening to labordeta} we get:

\begin{sat}\label{sat corrsum probability}
Let $M\subset\BR^s$ be a closed submanifold with reach $\tau(M)\ge 1$, pick two scales $0<\epsilon_2<\epsilon_1$. Also, for $n\ge 2$ set
$$\rho=
\sum_{i=1,2}\left(\frac 2{(n-1)^2}\cdot\frac{\vol(M)}{\CV(\epsilon_i)}+\frac 1{n-1}\cdot\left(\CR(\epsilon_i)-1\right)^2\right)$$
Then we have
$$\Prob\left(\left\vert\frac{\log X^{\epsilon_1}_{n}-\log X^{\epsilon_2}_{n}}{\log\epsilon_1-\log\epsilon_2}-\frac{\log\left(\frac{\vol(DM(\epsilon_1))}{\vol(DM(\epsilon_2))}\right)}{\log\epsilon_1-\log\epsilon_2}\right\vert>\Delta\right)\le \frac \rho{\left(1-\left(\frac{\epsilon_2}{\epsilon_1}\right)^{\frac\Delta 2}\right)^2}$$
for any $\delta>0$. \qed
\end{sat}

Let us get an slightly more user friendly version:

\begin{kor}\label{kor more humane}
With the same assumptions and notation as in Theorem \ref{sat corrsum probability} suppose that for some positive $\alpha_i$'s with $\alpha_1+\alpha_2=1$ and for some $\rho>0$ we have
$$n\ge 1+\frac 1{\alpha_i\cdot\rho}\cdot\left(\CR(\epsilon_i)-1\right)^2+\sqrt{\frac 2{\alpha_i\cdot\rho}\cdot\frac{\vol(M)}{\CV(\epsilon_i)}}$$
for $i=1,2$. Then we also have
$$P\left(\left\vert\frac{\log X^{\epsilon_1}_{n}-\log X^{\epsilon_2}_{n}}{\log\epsilon_1-\log\epsilon_2}-\frac{\log\left(\frac{\vol(DM(\epsilon_1))}{\vol(DM(\epsilon_2))}\right)}{\log\epsilon_1-\log\epsilon_2}\right\vert>\Delta\right)\le \rho\cdot\left(1-\left(\frac{\epsilon_2}{\epsilon_1}\right)^{\frac\Delta 2}\right)^{-2}$$
\end{kor}
\begin{proof}
In terms of Theorem \ref{sat corrsum probability} what we have to do is to guarantee for $i=1,2$ that 
$$\alpha_i\rho\ge\frac 2{(n-1)^2}\cdot\frac{\vol(M)}{\CV(\epsilon_i)}+\frac 1{n-1}\cdot\left(\CR(\epsilon_i)-1\right)^2$$
It thus suffices to ensure that $n-1$ is larger than the solution $X$ of the equation
$$a\stackrel{def}=\alpha_i\rho=\frac 1{X^2}\cdot\frac{2\cdot\vol(M)}{\CV(\epsilon_i)}+\frac 1{X}\cdot\left(\CR(\epsilon_i)-1\right)^2\stackrel{def}=\frac 1{X^2}c+\frac 1Xb$$
This is now a quadratic equation with positive solution 
$$X=\frac ba+\sqrt{\frac ca}$$
The claim follows.
\end{proof}

Again, if $M$ is such that all balls $\BB^M(x,\epsilon)$ have constant volume, then one can replace the first displayed equation in Corollary \ref{kor more humane} by
$$n\ge 1+\sqrt{\frac 2{\alpha_i\cdot\rho}\cdot\frac{\vol(M)}{\CV(\epsilon_i)}}.$$

\subsection{Searching decent scales}\label{sec: searche decent scales}

Given the dimension $d$ and the volume $\Vol(M)$, what are the optimal scales to run \eqref{eq corrsum intro} so that we have $\ge 90\%$ success probability? Let's see how we could find, if not the optimal scales, at least decent ones. First, for $1>\epsilon_1>\epsilon_2>0$ consider the quantity
$$\Delta=\Delta_{\epsilon_1,\epsilon_2}=\max\{\Delta_1,\Delta_2\}$$
where 
\begin{align*}
\Delta_1&\le\frac 12-\frac{\log\left(\frac{\int_0^{\sqrt 2\cdot 2\cdot\arcsin(\frac {\epsilon_1}2)}\sinh^{d-1}(t)\ dt}{\int_0^{\sqrt 2\cdot {\epsilon_2}}\sinh^{d-1}(t)\ dt}\right)}{\log\frac{\epsilon_1}{\epsilon_2}}+d\\
\Delta_1&\le\frac 12+\frac{\log\left(\frac{\epsilon_1}{2\cdot\arcsin(\frac {\epsilon_2}2)}\cdot\left(\frac{\sin(\epsilon_1)}{\sin(2\cdot\arcsin\frac {\epsilon_2}2)}\right)^{d-1}\right)}{\log\frac{\epsilon_1}{\epsilon_2}}-d
\end{align*}
From Lemma \ref{lem listening to labordeta} and Theorem \ref{sat pain in the ass} we get for all $n$ that 
$$\left\vert \frac{\log \frac{E(X_n^{\epsilon_1})}{E(X_n^{\epsilon_2})}}{\log\epsilon_1-\log\epsilon_2}
-d\right\vert
=\left\vert \frac{\log \frac{\vol(DM(\epsilon_1))}{\vol(DM(\epsilon_2))}}{\log\epsilon_1-\log\epsilon_2}-d\right\vert\le\frac 12-\Delta$$
Note that Theorem \ref{sat pain in the ass} has the condition $R<1$, but for dimension 1, it be replaced by the condition $R<2$, as explained in the proof of Corollary \ref{kor volume of balls}.

It follows that, as long as $\Delta>0$, if we take a very large number of points $n$ then we get that it is very likely that the \eqref{eq corrsum intro} returns the value $d$. Now, how many points we do actually need if we want to guarantee a $90\%$ rate of success? Well, with notation as in Theorem \ref{sat corrsum probability} we start by setting 
\begin{equation}\label{eq condition rho}
\rho=\rho_{\epsilon_1,\epsilon_2}=\frac 1{10}\cdot\left(1-\left(\frac{\epsilon_2}{\epsilon_1}\right)^{\frac\Delta 2}\right)^2
\end{equation}
Now, once we have $\rho$ we get from Corollary \ref{kor more humane} that if we take $\alpha\in(0,1)$, set $\alpha_1=\alpha$ and $\alpha_2=1-\alpha$, and if we take at least
\begin{align*}
n(\epsilon_1,\epsilon_2,\alpha,\vol(M))
&=\max_{i=1,2}\left(1+\frac 1{\alpha_i\cdot\rho}\cdot\left(\CR(\epsilon_i)-1\right)^2+\sqrt{\frac 2{\alpha_i\cdot\rho}\cdot\frac{\vol(M)}{\CV(\epsilon_i)}}\right)\\
&\le 1+\max_{i=1,2}\left(\frac 1{\alpha_i\cdot\rho}\cdot\left(\CR(\epsilon_i)-1\right)^2\right)+\\
&\phantom{BLABLABLA}+\left(\max_{i=1,2}\sqrt{\frac 2{\alpha_i\cdot\rho\cdot\CV(\epsilon_i)^2}}\right)\cdot\vol(M)^{\frac 12}
\end{align*}
points, then 
\begin{equation}\label{eq prob 10 percent}
\Prob\left(\left\vert \frac{\log \frac{X_n^{{\epsilon_1}}}{X_n^{{\epsilon_2}}}}{\log\epsilon_1-\log\epsilon_2}
-d\right\vert<\frac 12 d\right)>90\%.
\end{equation}
If we are interested in manifolds with $\Vol(M)\le V$ then \eqref{eq prob 10 percent} holds as long as we take at least
\begin{equation}\label{eq number of points}
\min_{\tiny\begin{array}{l}1>\epsilon_1>\epsilon_2>0\\ \text{with }\Delta_{\epsilon_1,\epsilon_2}>0\end{array}}\min_{\alpha\in(0,1)}n(\epsilon_1,\epsilon_2,\alpha)
\end{equation}
points and we use \eqref{eq corrsum intro} with constants $1>\epsilon_1>\epsilon_2>0$. Now, to find decent scales we can now minimize \eqref{eq number of points}. A program which numerically approximates that is available at \cite{program}. 

In fact, running also the program in each $d=1,2,\dots,10$ for the volume of the corresponding $d$-dimensional torus we get that $(\epsilon_1,\epsilon_2,\alpha)$ as in Table \ref{table scales} give smallish values for \eqref{eq number of points}. If we plug these constants in the formula for $n(\epsilon_1,\epsilon_2,\alpha,\vol(M))$ that we gave above we recover the statement of Theorem \ref{sat humane formula for number of points} stated in the introduction:

\begin{table}
\begin{tabular}{|l|l|l|l|l|l|l|l|l|}
\cline{1-4}\cline{6-9}
$d$ & $\epsilon_1$ & $\epsilon_2$ & $\alpha_1$ & & $d$ & $\epsilon_1$ & $\epsilon_2$ & $\alpha_1$\\ \cline{1-4}\cline{6-9}
 1 & $1.5$ & $0.19$ & $0.15$ & \phantom{BLABLABLA} &  6  & $0.4$ & $0.21$ & $0.03$ \\ \cline{1-4}\cline{6-9}
 2  & $0.78$  & $0.2$ & $0.11$ &  & 7 & $0.36$ & $0.21$ & $0.03$\\ \cline{1-4}\cline{6-9}
 3 & $0.63$ & $0.23$ & $0.09$ & &  8 & $0.33$ & $0.2$ & $0.02$ \\ \cline{1-4}\cline{6-9}
 4  & $0.54$& $0.23$ & $0.06$ & &  9 & $0.31$ & $0.19$ & $0.02$ \\ \cline{1-4}\cline{6-9}
 5 & $0.46$ & $0.22$ & $0.04$ & &  10 & $0.29$ & $0.18$ & $0.01$ \\ \cline{1-4}\cline{6-9}
\end{tabular}
\medskip
\caption{Decent scales for $\vol=\vol(\BT^d)$ in dimension $d=1,2,\dots,10$. It is evident that the values in Table \ref{table scales} can change if instead of using $\vol(\BT^d)$ as an input one chooses any other value. However, for whatever it is worth, if instead one chooses $10\cdot\vol(\BT^d)$ or even $100\cdot\vol(\BT^d)$ then nor much changes: in small dimensions (that is, up to dimension 3) the scales increase a bit, but for dimensions at least $4$ nothing changes.
}\label{table scales}
\end{table}

\begin{named}{Theorem \ref{sat humane formula for number of points}}
For $d=1,\cdots,10$ let $\epsilon_1$ and $\epsilon_2$ be scales as in the table below. Also, given a closed $d$-dimensional manifold $M\subset\BR^s$ with reach $\tau(M)\ge 1$ let $n$ be also as in the following table:
\begin{table}[H]
\begin{tabular}{|l|l|l|l|}
\hline
$d$ & $\epsilon_1$ & $\epsilon_2$ & n \\ \hline
$1$ & $1.5$ & $0.19$ & $9+21\cdot\vol(M)^{\frac 12}$ \\ \hline
 $2$  & $0.78$  & $0.2$ & $94+58\cdot\vol(M)^{\frac 12}$ \\ \hline
 $3$ & $0.63$ & $0.23$ & $635+146\cdot\vol(M)^{\frac 12}$ \\ \hline
 $4$  & $0.54$& $0.23$ & $2786+392\cdot\vol(M)^{\frac 12}$ \\ \hline
 $5$ & $0.46$ & $0.22$ & $7013+1119\cdot\vol(M)^{\frac 12}$ \\ \hline
 $6$  & $0.4$ & $0.21$ & $13221+3366\cdot\vol(M)^{\frac 12}$ \\ \hline
 $7$ & $0.36$ & $0.21$ & $25138+10644\cdot\vol(M)^{\frac 12}$ \\ \hline
 $8$ & $0.33$ & $0.2$ & $50033+34890\cdot\vol(M)^{\frac 12}$ \\ \hline
\end{tabular}
\end{table}
\begin{table}[H]
\begin{tabular}{|l|l|l|l|}
\hline
$d$ & $\epsilon_1$ & $\epsilon_2$ & n \\ \hline
 $9$ & $0.31$ & $0.19$ & $63876+119533\cdot\vol(M)^{\frac 12}$ \\ \hline
 $10$ & $0.29$ & $0.18$ & $139412+425554\cdot\vol(M)^{\frac 12}$ \\ \hline
\end{tabular}
\end{table}
Then, if we sample independently and according to the riemannian volume form a subset $X\subset M$ consisting of at least $n$ points, then we have 
$$\dim_{\Corr(\epsilon_1,\epsilon_2)}(X)=d$$
with at least $90\%$ probability.\qed
\end{named}

In the next section we discuss some (much smaller) heuristic bounds, discuss some numerical experiments, and compare with the performance of other estimators.

	\section{Heuristics}\label{sec: heur}

	Theorem \ref{sat corrsum probability} gives us a bound for the number of points needed in a data set to be able to get  from \eqref{eq corrsum intro} at least $90\%$ of the time its dimension. In concrete examples, we expect that this confidence level can be achieved with significantly less points. We will discuss this difference between theory and practice, suggesting a simpler heuristic model supported by computational examples.
	
\subsection*{Heuristic bound}
	We begin by discussing a heuristic model representing an ideal situation without curvature. More concretely we will be running \eqref{eq corrsum intro} at some scales $\epsilon_1$ and $\epsilon_2$ at which we can ignore curvature effects. In effect, we will act as if all balls in $M$ of radius at most $\epsilon_1$ were euclidean and totally geodesic. 

The first observation is that the statistic \eqref{eq corrsum intro} is computed from information extracted from the distances $\vert x-y\vert$ for those (unordered) pairs 
$$\{x,y\}\in PX(\epsilon)=\big\{\{x,y\}\subset X\text{ with }x\neq y\text{ with }\vert x-y\vert\le\epsilon\big\}$$
rather than from the points themselves. This means that the performance of the algorithm should depend on $\vert PX(\varepsilon_1)\vert$ and on the dimension $d$, instead of directly on the total number $n$ of points.

	We think of the distance $\vert x-y\vert$ for $\{x,y\}\in PX(\epsilon_1)$ as a random variable, and from now on, we will take the point of view that we have $N=\vert PX(\epsilon_1)\vert$ random variables $(X_i)_{1\leqslant i\leqslant N}$ given by taking the distance between pairs of points at distance at most $\varepsilon_1$ and sampled uniformly on $M$. We then consider the variables $Y_i$ equal to 1 if $X_i$ is smaller than $\varepsilon_2$ and 0 otherwise. With this notation in place, the estimator \eqref{eq corrsum intro} becomes
$$\dim_{\Corr(\epsilon_1,\epsilon_2)(X)}=\round\left(\frac{\log(\frac 1N\sum_{i=1}^NY_i)}{\log(\epsilon_2)-\log(\epsilon_1)}\right)$$
Note that since we are assuming that all the balls are euclidean, the mean value of $Y_i$ is then $E_d=(\varepsilon_2/\varepsilon_1)^d$ and as the variables $Y_i$ only take the values 0 and 1, their variance is $\sigma_d^2=E_d-E_d^2$.

In reality, the variables $X_i$, and thus the variables $Y_i$ have no reason to be independent. Still, most of them are when the volume is large when compared to the size of the data set. So, from now on, we will put ourselves in the ideal situation that the $N$ variables $Y_i$ are independent. Independence implies that the distribution of the sample mean $Z=\frac 1N\sum_i Y_i$ is binomial of parameters $N$ and $E_d$, which can be approximated by a normal distribution of mean value $E_d$ and of variance $\frac 1N\sigma_d^2$ using the central limit theorem. It is then known that the probability of $Z$ to be in the interval $[E_d-1.64\cdot\sigma_d/\sqrt{N}, E_d+1.64\cdot\sigma_d/\sqrt{N}]$ is about 90\%. If we set $\gap_d=\min\left(E_{d-0.5}-E_d, E_d-E_{d+0.5}\right)$, we want to find $N$ so that $1.64\cdot\sigma_d/\sqrt{N}=\gap_d$. This number gives a number of pairs sufficient to obtain the right dimension with a confidence of 90\%.
\medskip

Suppose for example that the manifold $M$ has dimension 4. For $\varepsilon_1=0.54$ and $\varepsilon_2=0.23$, the scales coming from Theorem \ref{sat humane formula for number of points}, we can then compute the values of $\gap_4$ and of $\sigma_4$ and deduce the required value of $N$.
	\begin{align*}
	\gap_4&=\min(E_{3.5}-E_{4},E_{4}-E_{4.5})\simeq 0.01143\\
	\sigma_4&=\sqrt{E_4-E_4^2}\simeq 0.1784\\
	N &=\left(1.64\cdot \sigma_4/\gap_4\right)^2\simeq 655 \text{ pairs}
	\end{align*}
	This is an approximation using the central limit theorem (that is, replacing the binomial distribution by the normal distribution), but more precise computation can be done working directly with the binomial distribution. Doing this, we can take $N$ down to 516 (see \cite{program}). This reasoning can be applied to obtain the required number of pairs for each dimension---the results are summerized in Table \ref{table heur bounds}. 
	\begin{table}[H]
		\begin{tabular}{|l|l|l|l|l|}
			\hline
			$d$ & $\epsilon_1$ & $\epsilon_2$ & {\bf $N$ for 90\%} & {\bf $N$ for 70\%}  \\ \hline
1 & $1.5$ & $0.19$ & $30$ & $10$ \\ \hline
 2  & $0.78$  & $0.2$ & $122$ & $40$ \\ \hline
 3 & $0.63$ & $0.23$ & $249$ & $111$ \\ \hline
 4  & $0.54$& $0.23$ & $516$ & $238$ \\ \hline
 5 & $0.46$ & $0.22$ & $878$ & $360$ \\ \hline
 6  & $0.4$ & $0.21$ & $1329$ & $554$ \\ \hline
 7 & $0.36$ & $0.21$ & $1719$ & $698$ \\ \hline
 8 & $0.33$ & $0.2$ & $2481$ & $1070$ \\ \hline
 9 & $0.31$ & $0.19$ & $3900$ & $1604$ \\ \hline
 10 & $0.29$ & $0.18$ & $5849$ & $2414$ \\ \hline
		\end{tabular}
\caption{Heuristic bounds for the size of $PX(\epsilon_1)$ needed to have $90\%$ and $70\%$ rate of success when applying \eqref{eq corrsum intro} to data sets sampled from a reach $1$ manifold.}\label{table heur bounds}
	\end{table}
We compare next these bounds with those in Theorem \ref{sat humane formula for number of points} and then discuss a few numerical experiments.

\subsection*{Comparisson between theoretical and heuristic bounds} A problem when comparing the heurestic bounds in Table \ref{table heur bounds} and those in Theorem \ref{sat humane formula for number of points} is that the former ones are given interms of the cardinality of $PX(\epsilon_1)$ while the latter ones are given in terms of the cardinality of $X$.

To connect these two quantities recall that we are acting as if all $\epsilon_1$-balls in $M$ were euclidean. Now, we get for example from Lemma \ref{lem listening to labordeta} that the number $\vert PX(\varepsilon_1)\vert$ of unordered pairs of points at distance at most $\varepsilon_1$ is approximately given by the formula
	\begin{equation}\label{pairs}
	\vert PX(\varepsilon_1)\vert\simeq \dfrac{n(n-1)}{2}\dfrac{\vol(B^{\BR^d}(\varepsilon_1))}{\vol(M)}.
	\end{equation}
Acting as if \eqref{pairs} were to give a perfect relation between the number of points and that of pairs, we get that to get the estimated $516$ pairs in the case that $M=\BT^4$ is the 4-dimensional torus we need $1958$ data points. In comparison, Theorem \ref{sat humane formula for number of points} gives an upper bound of $18262$ points.

Arguing like this we can convert the heuristic bounds in Table \ref{table heur bounds} to bounds for the needed cardinality of a data set in terms of the dimension and the volume of the underlying manifold---see Table \ref{table comparisson ns}.

	\begin{table}[h]
		\begin{tabular}{|l|l|l|}
			\hline
			$d$ & {\bf heuristic $n$} & {\bf $n$ from Theorem \ref{sat humane formula for number of points}}\\ \hline
1 & $5\cdot\vol(M)^{\frac 12}$ & $9+21\cdot\vol(M)^{\frac 12}$ \\ \hline
 2  & $12\cdot\vol(M)^{\frac 12}$ & $94+58\cdot\vol(M)^{\frac 12}$ \\ \hline
 3 & $22\cdot\vol(M)^{\frac 12}$ & $635+146\cdot\vol(M)^{\frac 12}$ \\ \hline
 4  & $50\cdot\vol(M)^{\frac 12}$ & $2786+392\cdot\vol(M)^{\frac 12}$ \\ \hline
 5 & $128\cdot\vol(M)^{\frac 12}$ & $7013+1119\cdot\vol(M)^{\frac 12}$ \\ \hline
 6  & $355\cdot\vol(M)^{\frac 12}$ & $13221+3366\cdot\vol(M)^{\frac 12}$ \\ \hline
 7 & $964\cdot\vol(M)^{\frac 12}$ & $25138+10644\cdot\vol(M)^{\frac 12}$ \\ \hline
 8 & $2949\cdot\vol(M)^{\frac 12}$ & $50033+34890\cdot\vol(M)^{\frac 12}$ \\ \hline
 9 & $9458\cdot\vol(M)^{\frac 12}$ & $63876+119533\cdot\vol(M)^{\frac 12}$ \\ \hline
 10 & $33021\cdot\vol(M)^{\frac 12}$ & $139412+425554\cdot\vol(M)^{\frac 12}$ \\ \hline
		\end{tabular}
\caption{Comparisson between the heuristic bound and the bound in Theorem \ref{sat humane formula for number of points} for the number of points that suffice to have $90\%$ success rate when applying \eqref{eq corrsum intro} to data sets sampled from a reach $1$ manifold.}\label{table heur bounds2}\label{table comparisson ns}
	\end{table}

One should keep in mind that entries the middle column in Table \ref{table comparisson ns} are only meningful for $\vol(M)$ large. Still, there is a very clear difference between both bounds, the heuristic bound and that from Theorem \ref{sat humane formula for number of points}. This difference is at least in part due to the curvature of the submanifold $M$, but things are not helped by either all the nested inequalities leading to Theorem \ref{sat humane formula for number of points} or the fact that the Bienaym\'e-Chebyshev inequality is not very precise.
	
\subsection*{Numerical evidence}

Experimental examples seem to confirm the heuristical bounds presented on Table \ref{table heur bounds}. By numerically sampling points on different manifolds, we obtain results close to what was expected. Here is the procedure we followed:

\begin{enumerate}
\item Choose a manifold $M$ of known dimension $d$ and consider the scales and number of pairs $N$ given by Table \ref{table heur bounds}.
\item Uniformly and independently sample points in $M$ until we obtain $N$ pairs at distance $\varepsilon_1$. The sampling is done by repeating a program specific to the desired manifold that samples a single point randomly and uniformly on it.
\item Estimate the dimension using estimator \eqref{eq corrsum intro} with scales $\varepsilon_1$ and $\varepsilon_2$.
\item Repeat steps 2 and 3 one hundred times and count the number of success.
\end{enumerate}

We ran this experiment a variety of manifolds of reach 1 and always obtained a rate of success in a range of $\pm 6\%$ of the target rate. This difference between the actual rate and the target would be totally normal even in an ideal situation. Indeed, repeating an experiment 100 time with a probability of success of 90\% gives a result with a standard deviation of 3 successes. Every experimental result then falls into the usual range of two times the standard deviation.

The manifolds were chosen to observe different situations:

\begin{itemize}
\item \textbf{Worms:} $0$-level set of a randomly produced function on $\BR^2$. The precise algorithm used for generating these manifolds and sampling from them is available in \cite{program}, as well as the algorithms for the other manifolds.
\item \textbf{Rotation torus:} Rotate around the $z$-axis the circle in the $xy$-plane of radius 1  and center $(2,0,0)$. This surface has reach $1$ and there is a mix of positive and negative curvature.
\item \textbf{Clifford torus:} The product $\BT^d=\BS^1\times\dots\times\BS^1\subset\BR^{2d}$ of $d$-circles of radius $1$. The Clifford torus is curved in Euclidean space but is flat as a Riemannian manifold.
\item \textbf{Flat torus:} This is an ideal situation. We consider namely the abstract manifold $\BR^d/2\pi\cdot\BZ^d$ with its inner distance---it is not embedded in some larger Euclidean space.
\item \textbf{Swiss roll:} The Swiss roll is one of the standard objects on which manifold learning argorithms seem to be tested, but it also adds a manifold with boundary to our list. For the sampling, we used the function \textit{make\_swiss\_roll} from the library \textit{scikit-learn}.
\item \textbf{Schwarz $P$ surface:} This is the triply periodic surface in $\BR^3$ with equation $\cos(x)+\cos(y)+\cos(z)=0$---it approximates one of Schwartz's triply periodic minimal surfaces and thus shows features of negative curvature, both intrinsic and extrinsic. To be able to deal with a finite volume surface we consider it as a submanifold of the 3-dimension flat torus.
\item \textbf{Spheres:} This is the standard sphere $\BS^d=\{x\in\BR^{d+1}\text{ with }\Vert x\Vert=1\}$. Spheres have reach $1$ and are positively curved.
\item \textbf{Gaussian distribution:} The standard Gaussian distribution in $\BR^d$. The reason why we test this in particular is to include a non-uniform distribution in our list.
\end{itemize}
The way we sample points depends on the concrete manifold under consideration, but we stress that we are sampling each point independently. More precisely, we are not aiming at getting point at some uniform distance of each other. 
\begin{figure}[h]
	\centering
	\includegraphics[width=10cm]{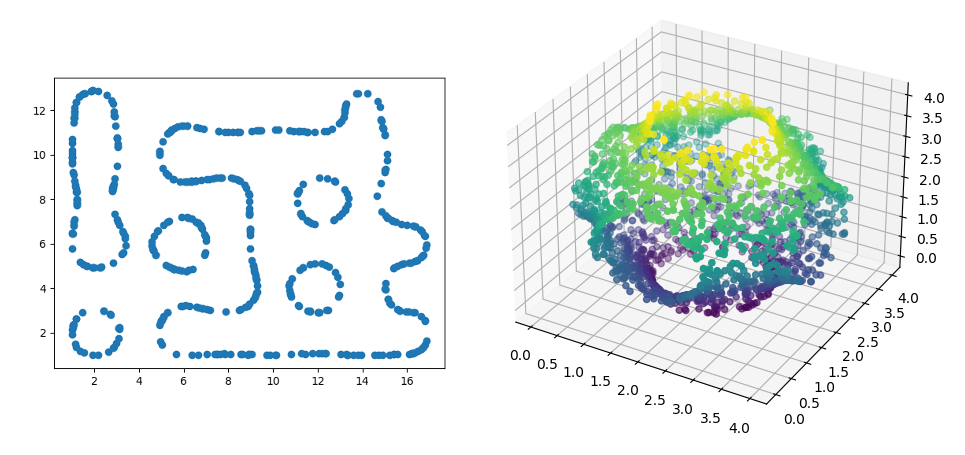}
	\caption{Examples of samplings from manifolds. 500 points on a "\textit{Worms}" manifold (Left) and 2\,000 points on a piece of the Schwarz surface (Right).}
	\label{samplings}
\end{figure}

\begin{table}[H]
	\begin{tabular}{|l|l|l|l|}
		\hline
		\textbf{Manifold} & $d$ & \textbf{90\% target} & \textbf{70\% target} \\
		\hline
	 	Worms                 & 1 & 88\% & 66\% \\ \hline
	 	Rotation torus                 & 2 & 92\% & 70\% \\ \hline
	 	Clifford torus        & 2 & 89\% & 69\% \\ \hline
	 	Flat torus            & 2 & 88\% & 66\% \\ \hline
	 	Swiss roll            & 2 & 93\% & 69\% \\ \hline
		Schwarz surface       & 2 & 88\% & 66 \% \\ \hline
	 	3-sphere              & 3 & 92\% & 76\% \\ \hline
	 	4-sphere              & 4 & 89\% & 75\% \\ \hline
	 	Product of two rotation tori   & 4 & 92\% & 70\% \\ \hline
	 	Clifford torus        & 4 & 93\% & 72\% \\ \hline
	 	Flat torus            & 4 & 90\% & 74\% \\ \hline
	 	Product of two Schwarz surfaces       & 4 & 92\% & 72\% \\ \hline
	 	Gaussian distribution in $\BR^4$ & 4 & 90\% & 76\% \\ \hline
	 	5-sphere              & 5 & 93\% & 74\% \\ \hline
	\end{tabular}
	\caption{Experimental rates of success for different manifolds by considering the ammount of pairs of points suggested by Table \ref{table heur bounds}.}\label{table expe}
\end{table}

These results tend to confirm the conclusions of the heuristic model. We can test this model a bit more by testing the estimator with the number of points given by Table \ref{table comparisson ns}. This is what we did in Table \ref{table expe 2}, and we still obtain results close to what was expected.

\begin{table}[H]
	\begin{tabular}{|l|l|l|l|}
		\hline
		\textbf{Manifold} & $d$ & \textbf{rate of success} \\
		\hline
		Clifford torus        & 2 & 91\% \\ \hline
		3-sphere              & 3 & 91\% \\ \hline
	 	Flat torus            & 4 & 91\% \\ \hline
	 	Product of two tori   & 4 & 94\% \\ \hline
	\end{tabular}
	\caption{Experimental rates of success for different manifolds by sampling the number of points given by the "heuristic" column in Table \ref{table comparisson ns}.}\label{table expe 2}
\end{table}

\subsection*{Reach free estimator}

A problem with all the results we have been discussing so far is that in practice one has little clue what the reach of the underlying manifold could be. However, as we already mentioned in the introduction, one can actually derive from Table \ref{table heur bounds} an estimator which does not need any a priori bound on the reach. 
\begin{quote}
{\bf Assumption:} We have a data set $X\subset\BR^s$ of which we think that it has been sampled from some mysterious submanifold $M\subset\BR^s$. We trust however that our data set is good enough and we want to test if $M$ could plaussibly have dimension $d$.
\end{quote}
\begin{quote}
{\bf Test:} For the chosen $d$, let $\epsilon_1$, $\epsilon_2$ and $N$ be as in Table \ref{table heur bounds} (say from the 90\% column). Now take $R>0$ to be minimal with $\vert PX(R)\vert\ge N$ and set $r=\frac{\epsilon_2}{\epsilon_1}R$. Now check if $\dim_{\Corr(R,r)}(X)=d$. 
\end{quote}
The heuristic discussion above, as well as the numerical experiments, suggest that as long as our data set is rich enough so that $R\le 0.54\cdot\tau(M)$ then we should get $\dim_{\Corr(R,r)}(X)=d$ with a 90\% probability.

Wondering what would happen if we run this algorithm on more "real" data sets, we chose 3 different data sets of respective dimension 1, 2 and 3 (see \cite{program}). Each data sets consist in 200 grayscale pictures of the 3D-model \textit{Suzanne} with random rotations (see Figure \ref{monkey}). We respectively randomized 1, 2 and 3 Euler angles to obtain the desired dimensions. The pictures are 64 by 64 pixels large and can thus be represented as points in $\mathbb{R}^{4096}$. The results are presented in Table \ref{table monkey}. For each data set, we test dimension 1, 2, 3 and 4 and we show the estimated dimension (before rounding it to the closest integer).

Note that when using the parameters of dimension 1 for estimating the dimension of a higher dimensional set, we usually get no points at distance $\varepsilon/2$, as $(1.5/0.19)^2$ is higher than $30$. When this happens, the estimator \eqref{eq corrsum intro} cannot be computed, we can only conclude that the dimension is probably bigger than 1.

\begin{table}[H]
	\begin{tabular}{|l|l|l|l|}
		\hline
		 \textbf{Hypothesis} & \textbf{Data set 1} & \textbf{Data set 2} & \textbf{Data set 3}\\
		\hline
	 	\textbf{dimension = 1} & \cellcolor[HTML]{FFDDDD}1.13 & $>1$ & $>1$ \\
	 	\hline
	 	\textbf{dimension = 2} & \cellcolor[HTML]{FFDDDD}1.59 & \cellcolor[HTML]{FFDDDD}2.02 & 3.54 \\
	 	\hline
	 	\textbf{dimension = 3} & 1.64 & 2.03  & \cellcolor[HTML]{FFDDDD}3.33 \\
	 	\hline
	 	\textbf{dimension = 4} & 1.45 & 2.08  & 3.40 \\
	 	\hline
	\end{tabular}
	\caption{Testing the dimension of the 3 "real" data sets. To test dimension=1 (resp. dimension=2, resp. dimension=3) we set $\epsilon_1$ so that $\vert PX(\epsilon_1)\vert=30$ (resp. 122, resp. 249). Cells with a pink background represent the tests in which the estimated dimension is consistent with the tested dimension. The cells with the value "$>d$" mean that there were no pairs at distance $\varepsilon_2$ for the corresponding number of pairs at distance $\varepsilon_1$.}\label{table monkey}
\end{table}

%

We can see in Table \ref{table monkey} that there were no Type I errors (meaning that we never rejected the true hypothesis). On the other hand we found a Type II error, when the data set of dimension 1 passed the test for dimension 2.

\section{Comparison with other estimators}\label{sec: compare}

In this paper we study the estimator \eqref{eq corrsum intro}, but a variety of other algorithms exist. For example, instead of doing statistics using only the distances between points, we could also use the angles between points, or other more complex features. These different approaches can be compared using the previous heuristic model. We refer to \cite{Camastra,CamastraS} and specially to \cite[Chapter 3]{LVbook} for a review of different dimension estimators. We will compare \eqref{eq corrsum intro} with the estimators ANOVA, local PCA, and with the implementation of \eqref{eq corrsum intro} where one tries to avoid picking scales, reading instead the dimension from a log-log chart.
	
	\subsection*{ANOVA}
	Diaz, Quiroz and Velasco propose in \cite{ANOVA} a method based on the angles refered as ANOVA in the literature. Their idea is to estimate the local dimension around a point $x$ by considering its $k$ nearest neighbors and the ${k \choose 2}$ angles at $x$ formed by these points. From the variance of these ${k \choose 2}$ angles, we can identify the closest $\beta_d$ and deduce the dimension $d$, where $\beta_d$ is defined as follows.
	\[\beta_d=\frac 1{\vol(\BS^{d-1})^2} \int_{\BS^{d-1}\times\BS^{d-1}} \left(\angle(\theta,\eta)-\dfrac{\pi}{2}\right)^2 d\theta d\eta\]
	The global dimension can then be recovered by taking the median, the mode or the mean of the local dimensions.
	
To be able to compare ANOVA to the estimator \eqref{eq corrsum intro} we will instead take the variance of all angles. More precisely, consider every ordered triple of points in which all three points are at distance at most $\varepsilon_1$ from each other. Each such triple $(x,y,z)$ determines an angle $\angle(y-x,z-x)$. We compute the variance of the so-obtained angles we can locate the closest $\beta_d$. We take that $d$ to be the ANOVA dimension of our data set.

	Putting ourselves again in the ideal situation that we are working in a scale at which curvature can be ignored and assuming (and that is a lot of assuming in this case) that the angles we find are independent of each other, we can argue as earlier in the discussion of the heuristic bound and we get that, in dimension 4, we would need at least $652$ angles to achieve a confidence level of 90\% with the ANOVA estimator. 

Suppose now that we sample 1958 points from the 4-dimensional Clifford torus $M=\BT^4$. That number was chosen so that we expect to have $516$ pairs within $\epsilon_1=0.54$ of each other, the heuristic bound for \eqref{eq corrsum intro} in dimension $4$. On the other hand we expect to have 
$$\text{number of unordered triples}={1958 \choose 3}\cdot\left(\frac{\vol(B^{\BR^4}(\epsilon_1))}{\vol(M)}\right)^2\simeq 91.$$
Each unordered pair gives 3 angles, meaning can expect to find about $273$ angles. This is much less than what we estimated that would needed.

And this phenomenon gets worse when the volume of the underlying manifold grows: if the manifold has volume $64\pi^4$ and if we have a data set with $516$ pairs of points within $\epsilon_1=0.54$ then we expect to only have $45$ triples, that is about $135$ angles. The reason for this is that the number of points needed to have a given number of pairs grows with $\vol(M)^{\frac 12}$ while it grows as $\vol(M)^{\frac 23}$ when we fix the number of triples instead. This means that for large volumes the algorithm will only become worse as finding triples of points will become more and more difficult.\\

We numerically compared the rates of success of estimator \eqref{eq corrsum intro} and of ANOVA for the Clifford torus with the same scales. The results are presented in Table \ref{table anova}. This experiment shows that, for this example and for these scales, estimator \eqref{eq corrsum intro} gives significantly better results than ANOVA. However, we cannot conclude that estimator \eqref{eq corrsum intro} has better performance in general. In particular, for manifolds of smaller volume or for different choices of scales, ANOVA could in principle give better results.

\begin{table}[H]
	\begin{tabular}{|l|l|l|l|l|l|}
		\hline
		\textbf{Manifold} & $d$ & \textbf{number of points} & \textbf{estimator \eqref{eq corrsum intro}} & \textbf{ANOVA}\\
		\hline
		Clifford torus & 2 & 76 & 93\% & 65\% \\ \hline
		Clifford torus & 3 & 347 & 93\% & 67\% \\ \hline
	\end{tabular}
	\caption{Comparison of the experimental rates of success between estimator \eqref{eq corrsum intro} and ANOVA on the Clifford torus.}\label{table anova}
\end{table}

\begin{bem}
Recall that we have considered a variation of ANOVA---our conclusions should also apply to the original algorithm as well.
\end{bem}

	\subsection*{Local PCA}
	Principal Component Analysis aims to find the best linear space containing a given data set. According to \cite{BKSW}, PCA is the gold standard of dimension estimation. It works as follows. To our given data set $X=(x_1,\dots,x_n)\subset\BR^N$ we associate first the mean 
	$$\bar x=\frac 1n\sum_{i=1}^n x_n$$
	and then the $n\times N$ matrix $A$ whose rows are the vectors $u_i-\bar u$, and one computes then the singular values $s_1\ge s_2\ge\dots\ge s_{\min\{n,N\}}$. If $X$ is contained in a linear subspace of dimension $d$ then $s_k=0$ for all $k\ge d+1$. Accordingly, one can declare that $X$ has PCA-dimension $k$ if the gap $s_k-s_{k+1}$ is maximal. Another possibility would be to fix a threshold $\epsilon$ and declare the PCA-dimension of $X$ to be the largest $k$ with $s_k\ge\epsilon$. For example, if what one wants to do is to test the hypothesis that $X$ is contained in a $d$-dimensional subspace then one can check for example if $s_{d+1}$ is below some threshold $\epsilon<\frac 1{d+2}$---the number $\frac 1{d+2}$ arises because it is the expected value for $s_d$ if $X$ is uniformly sampled out of the unit ball in $\BR^d$.

	In any case, that was PCA. The idea of {\em local PCA}, or {\em Nonlinear PCA}, is to apply PCA to certain subsets of the data set and then, for good measure, average the so obtained numbers. As we see, one does not only need to agree on what one calls PCA, but also on what subsets does one wants to subject to the PCA treatement. For example, as in \cite{BKSW} one can cluster the data set using single linkage clustering\footnote{In single linkage clustering the clusters are, for some $\epsilon$, the connected components of the graph with vertex set $X$ and where two vertices are joined by an edge if they are with $\epsilon$ of each other.} and then apply PCA to each custer. This is, in our humble and uneducated opinion, a very reasonable choice if, for data sets $X$ uniformly sampled out of a submanifold $M\subset\BR^N$, what one wants to do is to find linear spaces (approximately) containing each connected component of $M$. On the other hand, if what one wants to do is to recover the dimension of $M$ then it seems reasonable to rather apply PCA to (some of) the sets $X\cap\BB(x,\epsilon)$ for $x\in X$ and for some $\epsilon$ chosen so that $M\cap\BB(x,\epsilon)$ is well-approximated by its tangent space.\\

Note now that every set consisting of $d+1$ points is contained in a $d$-dimensional affine subspace of $\BR^s$. This means that if we want to distinguish dimension $d$ from dimension $d+1$ we need at the very least to find $d+2$ tuples of nearby points. Now, if we once again sample $1958$ points out of the Clifford torus $M=\BT^4$ so that we can expect the magic number of $516$ paris of points at most at distance $0.54$ from each other, then we can expect to find
$$\begin{array}{c}
\text{number of unordered}\\
\text{6-tuples at scale }0.54
\end{array}={1958 \choose 6}\cdot\left(\frac{\vol(B^{\BR^4}(0.54))}{\vol(M)}\right)^5\simeq 0.110373....$$
In other words, the expectation is to not have any $6$-tuples, meaning that, if we are working at the same scales as we were implementing \eqref{eq corrsum intro}, then we have no chance of distinguishing our 4-manifold from a 5-manifold using local PCA. Note that even if we work at a scale of $2$, scale at which we are acting as if the round sphere were totally flat, then we can expect to only have 2 6-tuples
$$\begin{array}{c}
\text{number of unordered}\\
\text{6-tuples at scale }2
\end{array}={1958 \choose 6}\cdot\left(\frac{\vol(B^{\BR^4}(2))}{\vol(M)}\right)^5\simeq 2.043944....$$
As we see, even if we work at scales at which we are flat earthers, we do not have by far enough sufficiently populated clusters for local PCA to be meaningful.

\subsection*{log-log plots}
Another method to derive the estimator \eqref{eq corrsum intro} when the reach is unknown, is to take a lot of scales
$$\epsilon_1<\epsilon_2<\dots<\epsilon_k,$$
and to consider the graph of the piecewise linear function with corners
\[\big(\log(\epsilon_i),\log(\vert PX(\epsilon_i)\vert)\big)\]
and try to read out some sort of meaningful slope of that function (see \cite[Chapter 3]{LVbook}). We implemented this procedure for 1\,000 points uniformly sampled from the product of two 2-dimensional rotation tori of reach 1. We obtained the graph presented on Figure \ref{loglog}.

\begin{figure}[h!]
    \begin{center}
        \resizebox{.9\linewidth}{!}{\input{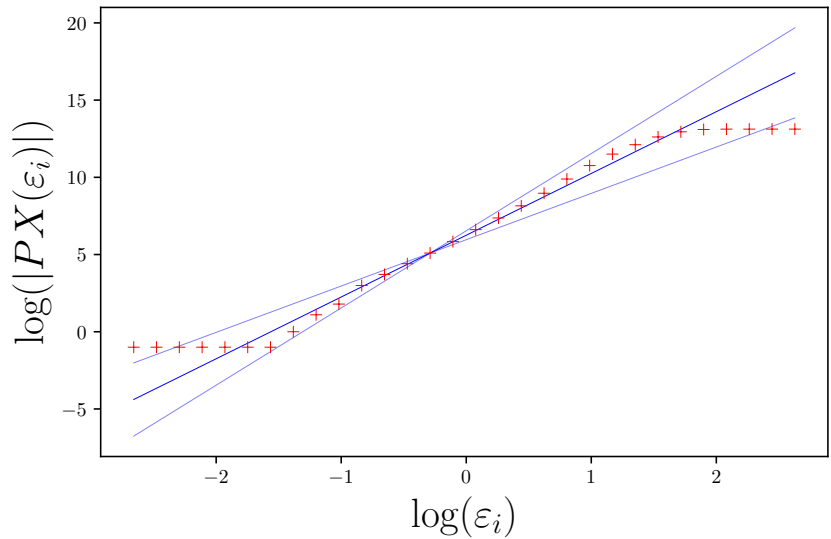}}
    \end{center}
    \caption{Log-log plot for a set of 1000 points sampled out of the product of two rotation tori. The graph also features three lines of slope 3, 4 and 5, for comparison.}
    \label{loglog}
\end{figure}

Three different parts can be observed on this graph:
\begin{enumerate}
\item A flat part, when the scale is smaller than the minimal distance between the points (represented with the value -1 on the y-axis). 
\item A mostly linear part whose slope should approximate the dimension of $M$.
\item A flat part when the scale becomes greater than the diameter of $M$ (which plateaus at $\log(1000\cdot 999/2)\simeq 13.12$).
\end{enumerate}
It is notable that in this example the second part looks linear even after twice the reach, even if the behavior of $\log(\vert PX(\epsilon_i)\vert)$ beyond this scale is unpredictable. Around the middle of this graph, the slope of the linear part seems to be close to 4, as predicted.
\medskip

We also applied this method from the 3 "real" datasets of pictures of \textit{Suzanne} with random rotations. The results are presented in Figure \ref{monkey loglog}. For dimension 1 and 2, it worked surprisingly well and the log-log plots clearly show what is the right dimension. The plot for the data set of dimension 3 is not as clear, but the result seems consistent.
\medskip

Let us conclude with two further comments on the log-log procedure:

\noindent{\bf (1)} The log-log procedure seems to allow us to estimate the dimension of $M$ without any assumption on its reach. It is however difficult to define a precise algorithm for it. Consequently, we cannot really compare its performance with that of the other estimators.

\noindent{\bf (2)}
Additionally, this method gives us insight on the choice of scales that could be relevant. Indeed, the scales determined in Section \ref{sec: searche decent scales} have no reason to be optimal for the heuristic model. And as this model does not take the curvature nor the reach into account, it is impossible to use it to find optimal scales. When considering Figure \ref{loglog}, it seems natural to choose two scales that are far apart to measure the slope as precisely as possible, but close enough so that they do no get too close to the extremities. Here, two scales that would seem relevant could be $\exp(0.5)\simeq 1.6$ and $\exp(-0.5)\simeq 0.6$. In any case, the choice of scales $0.54\simeq\exp(-0.6)$ and $0.23\simeq\exp(-1.5)$ seem to be far from optimal. Using the estimator \eqref{eq corrsum intro} with scales 1.6 and 0.6 on computational examples, we obtain a rate of success of 99\% with only 1\,000 points for the product of two tori. For the scales 0.54 and 0.23, the heuristic model gives us a number of points of 3916 for a rate of success of 90\%. This rate falls to 41\% with only 1\,000 points on computational examples. However, these scales have been chosen after the study of the results to obtain the desired answer on a given manifold. But this does not give us a method to choose optimal scales a priori. For example, the scales 1.6 and 0.6 one give a rate of success of 45\% for 100 points on the 4-sphere, while the scale 1 and 0.6 give a rate of 92\%.

\begin{figure}[H]
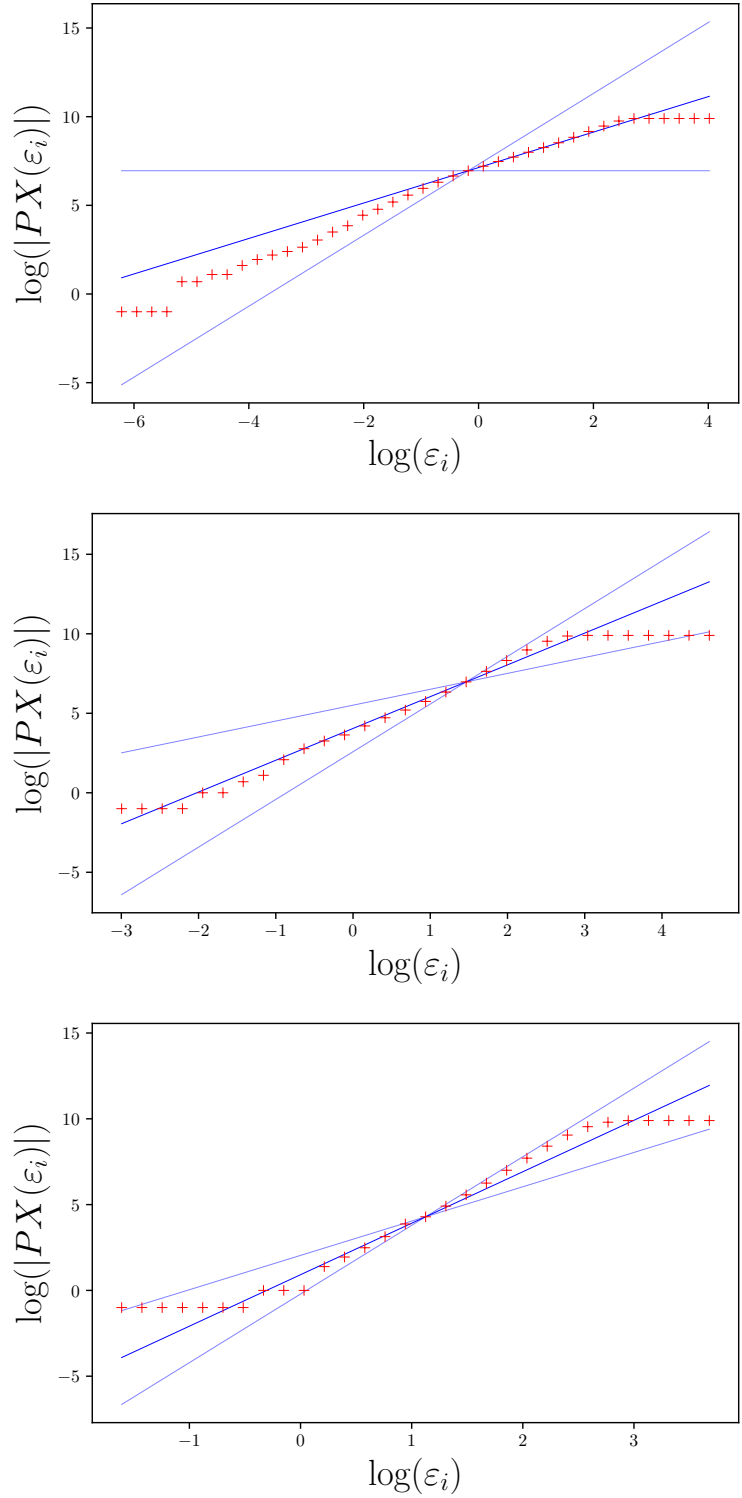

    \begin{center}
        \resizebox{.8\linewidth}{!}{\input{monkey1.pgf}}
        \resizebox{.8\linewidth}{!}{\input{monkey2.pgf}}
       	\resizebox{.8\linewidth}{!}{\input{monkey3.pgf}}
    \end{center}
    \caption{Log-log plots the 3 "real" datasets of pictures of \textit{Suzanne} with random rotations of dimension 1, 2 and 3. The blue lines have slopes of $d-1$, $d$ and $d+1$.}
    \label{monkey loglog}
\end{figure}


\begin{thebibliography}{[12]}

\bibitem{Aamari}
E. Aamari, J. Kim, F. Chazal, B. Michel, A. Rinaldo, and L. Wasserman, {\em Estimating the reach of a manifold}, Electron. J. Stat. 13 (2019).

\bibitem{AdamsandCo}
H. Adams, M. Aminian, E. Farnell, M. Kirby, J. Mirth, R. Neville, C. Peterson and C. Shonkwiler, {\em  A fractal dimension for measures via persistent homology}, in {\em Topological data analysis---the Abel Symposium 2018}, Abel Symp., 15, Springer 2020.

\bibitem{BNS}
M. Belkin, P. Niyogi, and V. Sindhwani. {\em Manifold regularization: a geometric framework for learning from labeled and unlabeled examples}, J. Mach. Learn. Res., (2006).

\bibitem{Berger}
M. Berger, {\em A panoramic view of Riemannian geometry}, Springer-Verlag, Berlin, 2003.

\bibitem{BG14}
J.-D. Boissonnat and A. Ghosh, {\em Manifold reconstruction using tangential Delaunay complexes}, Discrete Comput. Geom., 51, (2014).

\bibitem{Boissonat}
J.-D. Boissonnat, A. Lieutier, and M. Wintraecken, {\em The reach, metric distortion, geodesic convexity and the variation of tangent spaces}, J. Appl. Comput. Topol. 3, (2019).

\bibitem{BKSW}
P. Breiding, S. Kali\v{s}nik, B. Sturmfels and M. Weinstein, {\em Learning algebraic varieties from samples}, Rev. Mat. Complut. 31, (2018).

\bibitem{Bridson-Haffliger}
M. Bridson and A. Haefliger, {\em Metric spaces of non-positive curvature}, Grundlehren der mathematischen Wissenschaften  319. Springer-Verlag, 1999.

\bibitem{Camastra}
F. Camastra, {\em Data dimensionality estimation methods: a survey}, Pattern Recognition, Volume 36, (2003).

\bibitem{CamastraS}
F. Camastra and A. Staiano, {\em Intrinsic dimension estimation: advances and open problems}, Information Sciences 328, (2016).

\bibitem{Cheerger-Ebin}
J. Cheeger and D. Ebin, {\em Comparison theorems in Riemannian geometry}, AMS Chelsea Publishing, 2008. 

\bibitem{ANOVA}
M. D\'iaz, A. Quiroz, and M. Velasco, {\em Local angles and dimension estimation from data on manifolds}, J. Multivariate Anal. 173, (2019).

\bibitem{Eckmann-Ruelle}
J.P. Eckmann and D. Ruelle, {\em Fundamental liminations for estimating dimensions and lyapounov exponents in dynamical systems}, Physica D-56, (1992).

\bibitem{Federer curvature measures}
H. Federer, {\em Curvature measures}, Trans. Am. Math. Soc. 93, (1959).

\bibitem{Fefferman}
C. Fefferman, S. Mitter, and H. Narayanan, {\em Testing the manifold hypothesis}, J. Amer. Math. Soc. 29, (2016).

\bibitem{Fukunaga}
K. Fukunaga, {\em Intrinsic dimensionality extraction}, in {\em Handbook of Statistics}, Volume 2, 1982.

\bibitem{Grassberger}
P. Grassberger, {\em Do climatic attractors exist?}, Nature 323, (1986).

\bibitem{Grassberger-Procaccia}
P. Grassberger and I. Procaccia, {\em Measuring the strangeness of strange attractors}, Physica D9, (1983).

\bibitem{program}
L. Grillet, \href{https://github.com/lgrillet/dim-estimation}{https://github.com/lgrillet/dim-estimation}

\bibitem{LVbook}
J. Lee and M. Verleysen, {\em Nonlinear dimensionality reduction}, Information Science and Statistics, Springer 2010.

\bibitem{NM}
H. Narayanan and S. Mitter, {\em Sample complexity of testing the manifold hypothesis}, in {\em Advances in Neural Information Processing Systems 23}, Neural Information Processing Systems, 2011.

\bibitem{NSW}
P. Niyogi, S. Smale and S. Weinberger, {\em Finding the homology of submanifolds with high confidence from random samples}, Discrete Comput. Geom. 39, (2008).

\bibitem{Procaccia}
I. Procaccia, {\em Complex of just complicated?}, in Nature 333, (1998).

\bibitem{Simpelaere}
D. Simpelaere, {\em Correlation Dimension}, Journal of Statistical Physics volume 90, (1998).

\bibitem{Takens83}
F. Takens, {\em Invariants related to dimension and entropy}, in {\em Atas do 13 coloqkio brasileiro de mat\'ematica}, Rio de Janeiro, 1983.

\bibitem{Takens85}
F. Takens, {\em On the numerical determination of the dimension of an attractor}, in {\em Dynamical systems and bifurcations (Groningen, 1984)}, Lecture Notes in Math., 1125, Springer, 1985.

\bibitem{Theiler90}
J. Theiler, {\em Estimating fractal dimension}, J. Opt. Soc. Am. A 7, (1990).

\bibitem{Weinberger}
S. Weinberger, {\em The complexity of some topological interference problems}, Found. Comput. Math. 14, (2014).


\end{thebibliography}
\end{document}